%% file: main.tex
\documentclass[11pt,letter]{article}

\usepackage{comment} 
\usepackage{fullpage} 
\usepackage{amsmath}
\usepackage{amssymb}
\usepackage{bbm,bbold}
\usepackage{mathtools}
\usepackage{setspace}
\usepackage{graphicx}
\usepackage[dvipsnames]{xcolor}
\usepackage[colorlinks=true, allcolors=black]{hyperref}
\hypersetup{
  bookmarks=true,     
  bookmarksopen=true, 
}
\usepackage{algorithm}
\usepackage{thm-restate}
\usepackage{algorithmicx}
\usepackage{nicematrix}
\usepackage{algpseudocode}
\usepackage{amsthm}
\usepackage{thm-restate}
\usepackage{enumerate}
\usepackage[normalem]{ulem}
\usepackage{float}
\usepackage{afterpage}
\usepackage[numbers,square]{natbib}
\usepackage[nameinlink,capitalise]{cleveref}
\usepackage{subcaption}

\newcommand\blfootnote[1]{%
  \begingroup
  \renewcommand\thefootnote{}\footnote{#1}%
  \addtocounter{footnote}{-1}%
  \endgroup
}

\allowdisplaybreaks
\usepackage{pgfplots,tikz}
\pgfplotsset{compat=1.17}
\usetikzlibrary{calc,arrows.meta,patterns,pgfplots.fillbetween,decorations.pathmorphing}

\usepackage[]{color-edits}

\title{Breaking the Envy Cycle:\\ Best-of-Both-Worlds Guarantees for Subadditive Valuations}
\author{
Michal Feldman
\and
Simon Mauras
\and
Vishnu V. Narayan
\and
Tomasz Ponitka
}
\date{Tel Aviv University\\[0.1in]\today}

\addauthor{TP}{red}

\addauthor{MF}{blue}

\addauthor{VN}{violet!66!red}

\addauthor{SM}{green!50!black}

\addauthor{TODO}{brown}

\newtheorem{theorem}{Theorem}
\newtheorem{lemma}{Lemma}[section]
\newtheorem{proposition}[lemma]{Proposition}

\newtheorem{corollary}[lemma]{Corollary}
\theoremstyle{definition}
\newtheorem{definition}[lemma]{Definition}
\newtheorem{remark}[lemma]{Remark}
\newtheorem{example}[lemma]{Example}

\newcommand{\agents}{[n]}
\newcommand{\items}{[m]}

\newcommand{\reals}{\mathbb{R}}

\newcommand{\SD}{\succeq_{\mathrm{SD}}}
\newcommand{\SC}{\succeq_{\mathrm{SC}}}

\begin{document}
\maketitle

\begin{abstract}
We study best-of-both-worlds guarantees for the fair division of indivisible items among agents with subadditive valuations. Our main result
establishes the existence of a random allocation that is simultaneously ex-ante $\frac{1}{2}$-envy-free, ex-post $\frac{1}{2}$-EFX and ex-post EF1, for every instance with subadditive valuations.
We achieve this result by a novel polynomial-time algorithm that randomizes the well-established envy cycles procedure in a way that provides ex-ante fairness.
Notably, this is the first best-of-both-worlds fairness guarantee for subadditive valuations, even when considering only EF1 without EFX.
\end{abstract}
\section{Introduction}
\label{sec:intro}

In a fair division problem, the goal is to allocate a set $M$ of $m$ indivisible items among $n$ agents {\em fairly}. Every agent has a valuation function $v_i:2^M \rightarrow \reals^{\geq 0}$, mapping every bundle of items $S \subseteq M$ to a (non-negative) real number $v_i(S)$. An allocation $X=(X_1,\ldots,X_n)$ is a partition of the items among the agents, where $X_i$ is the bundle allocated to agent $i$.\blfootnote{This work was supported by the European Research Council (ERC) under the European Union's Horizon 2020 research and innovation program (grant agreement No. 866132), by an Amazon Research Award, and by the NSF-BSF (grant number 2020788).} 

The problem of allocating resources fairly dates back to Aristotle \cite{Chroust42}. 
Many mathematical notions of fairness have been considered in the literature (both for divisible and indivisible items).
A particularly natural fairness notion that generated significant interest in the literature 
is {\em envy-freeness} {(EF)}, requiring that every agent prefers her own bundle to any other agent’s bundle \cite{Fol67, V74}. 
That is, an allocation $X$ {is EF if} for every two agents $i$ and $j$, it holds that $v_i(X_i) \geq v_i(X_j)$. 
{The notion of envy-freeness extends to {\em random} allocations in two ways. We say that a random allocation is EF {\em ex-ante} (before the randomization is realized) if every agent prefers her own bundle to any other agent’s bundle {\em in expectation}. We say that a random allocation is EF {\em ex-post} (after the randomization is realized) if every deterministic allocation in its support is EF.}

Achieving ex-ante EF is quite easy. 
For example, allocating all items to a single agent {chosen uniformly at random} is trivially ex-ante EF. 
However, an ex-ante EF allocation may {still} be arbitrarily unfair {ex-post}. 
Indeed, in the example above, one agent receives all the items, and will surely be envied by all other agents. 
On the other hand, ex-post {EF} is too strong of a requirement, as even the simplest setting of a single item desired by two agents does not admit any EF allocation. 

\vspace{0.25cm}
\textbf{Best-of-both-worlds fairness.}
Recently, \citet{A19} posed, as an interesting new research direction, the question of finding random allocations that simultaneously achieve desirable ex-ante and ex-post properties.
Shortly thereafter, \citet{FSV20} studied this problem in the fair division domain, with the goal of obtaining a random allocation that is simultaneously ex-ante EF and ex-post ``relaxed EF''. The approach of constructing random allocations with strong ex-ante and ex-post fairness guarantees has since been known as the ``best-of-both-worlds'' approach.

In their work, \citet{FSV20} focus on the notion of EF1 --- envy-freeness up to one item~\cite{LMMS04,B11} --- which requires that the envy of every agent $i$ toward another agent $j$ can be removed by the elimination of at most one item from agent $j$'s bundle, {i.e., there exists an item $g\in X_j$ such that $v_i(X_i) \geq v_i(X_j \setminus \{g\})$.
Unlike EF, an EF1 allocation always exists \cite{B11}.
Moreover, \citet{FSV20} and \citet{A20} showed that any instance with additive valuations --- where an agent's value for a bundle is the sum of her values for the individual items --- admits a random allocation that is simultaneously ex-ante EF and ex-post EF1.
Namely, an envy-free distribution over deterministic allocations such that each satisfies EF1.

While this is a promising result, it is limited in two ways. 
First, it is restricted to additive valuations, while in most practical settings, agent valuations are non-additive. 
For example, items may exhibit substitutability, {where an agent might like to get one of two given items but have no additional value for having both of them}.
Second, this result holds with respect to EF1, but does not hold with respect to {another important} (and stronger) relaxation of envy-freeness, known as EFX --- envy-freeness up to {\em any} item~\cite{CKMPSW16} --- which requires that the envy of every agent $i$ toward another agent $j$ can be removed by the elimination of {\em any} item from agent $j$'s bundle.

To demonstrate why EFX may be more desirable than EF1, consider the following scenario. 
Suppose there are three items, $a,b,c$, and two agents with identical additive values of $100, 50, 50$ for items $a,b,c$, respectively.
Any reasonable fairness notion in this case would allocate item $a$ to one agent and items $b,c$ to the other agent, resulting in a value of $100$ to each agent. 
This allocation is the only EFX allocation for this instance. 
However, one can easily verify that allocating items $a,b$ to agent 1 (for a value of $150$) and item $c$ to agent 2 (for a value of $50$) is EF1: removing item $a$ from agent 1's bundle removes agent 2's envy.
{Indeed, extending best-of-both-worlds results to ex-post EFX appeared as one of the open problems in \citet{FSV20}.

While many attempts have been made towards improving the results of~\citet{FSV20} and~\citet{A20}, none of the subsequent works specifically resolve these issues. Instead, they either give best-of-both-worlds results that hold only for highly structured valuation classes, such as binary additive valuations (\citet{AAGW15,HPPS20}), matroid rank valuations (\citet{BEF21}), and multi-demand valuations (\citet{HSV23}), or study the case of arbitrary entitlements, where they obtain only weaker ex-post guarantees (\citet{AGM23,HSV23}).}}

\vspace{0.25cm}
\textbf{Our contribution.}
With the desire to go beyond additive valuations and beyond the EF1 notion, the key question we address in this paper is whether there always exists a random allocation that is simultaneously ex-ante EF and ex-post EFX, for any instance with subadditive valuations.

A valuation function is subadditive if $v(S)+v(T) \geq v(S \cup T)$ for every two bundles of items $S$ and $T$.
The class of subadditive valuations encompasses important subclasses, such as additive valuations and submodular valuations (namely, valuations exhibiting decreasing marginal values).

Let us consider first the simplest scenario of two agents with additive valuations. While not immediate, we show that this scenario always admits a random allocation that is simultaneously ex-ante EF and ex-post EFX (see \Cref{prop:additive-2-agents}).

But this positive result breaks as soon as we go either beyond two agents or beyond additive valuations. 
Going beyond two agents, even for additive valuations, the mere existence of ex-post EFX, without any additional requirement, is an intriguing open problem. 
{In fact, it is widely considered to be the biggest open problem in fair division~\cite{P20}\footnote{EFX is only known to exist for some special cases, including three additive agents~\cite{CGM20} or identical valuations~\cite{PR18}.}.}
Similarly, going beyond additive valuations, we show that even with only two agents and submodular valuations, there are instances that do not admit any allocation that is simultaneously ex-ante EF and ex-post EFX. 
This is demonstrated in \Cref{prop:2agent-submod-tight}.

{As {is} standard in the literature, to address this problem}
we turn to the framework of {\em approximate fairness}. 
A deterministic allocation is $\beta$-EFX if every agent prefers her own bundle to a $\beta$ fraction of any other agent's bundle after eliminating any single item from it. 
A random allocation is ex-post $\beta$-EFX 
if every allocation in its support is $\beta$-EFX. 
Similarly, a random allocation is ex-ante $\alpha$-EF if every agent prefers her own bundle to an $\alpha$ fraction of every other agent's bundle, in expectation.

The approach taken by \citet{FSV20} and \citet{A20} gives no ex-post $\beta$-EFX guarantees for any $\beta>0$ (see \Cref{sec:ps_explained}). Therefore, the question that drives us in this work is the following:

\vspace{0.1in}
\noindent
{\bf Main Question:}
Are there constants $\alpha$ and $\beta$ such that an allocation that is simultaneously ex-ante {$\alpha$}-EF and ex-post {$\beta$}-EFX is guaranteed to exist for every profile of subadditive valuations?

\vspace{0.1in}
Our main result is {an affirmative answer to this question}:

\vspace{0.1in}
\noindent
{\bf Main Result (\Cref{thm:main_theorem}):}
For every instance with subadditive valuations, \Cref{alg:fair_envy_cycles} outputs a random allocation that is ex-ante $\frac{1}{2}$-EF, ex-post $\frac{1}{2}$-EFX, and ex-post EF1 in polynomial time.

\vspace{0.1in}
This result constitutes a significant advancement in the best-of-both-worlds fairness literature in multiple ways:
\begin{enumerate}
\setlength{\itemsep}{1pt}
    \item We are the first to provide best-of-both-worlds results for {\em subadditive} valuations -- our work even extends the previous work on (EF, EF1) from additive all the way to subadditive, {and gains ex-post $\frac{1}{2}$-EFX in the process,} while losing only a factor of $2$ in the ex-ante envy. 
    \item {No better approximation of EFX than $\frac{1}{2}$-EFX is known to exist}, even for the smaller class of submodular valuations, and even without any ex-ante EF requirements. 
We match this guarantee for the broader class of subadditive valuations while providing ex-ante guarantees in addition.
    \item 
    To the best of our knowledge, we are the first to consider randomized envy cycles in a way that provides ex-ante EF guarantees.
    Envy cycles is one of the {most widely-used} procedures for discrete fair division. Indeed, the envy cycles algorithm was one of the earliest works in the fair division of indivisible items \cite{LMMS04}, and it has since been used to achieve many fairness results, including EF1 and EFX guarantees \cite{LMMS04,PR18}; see \Cref{sec:related-work} for more details.
\end{enumerate}

We complement our positive result with the following upper bounds (holding even for 2 agents):

\vspace{0.1in}
{
\noindent {\bf {Impossibility Result} (\Cref{prop:ub-subadditive}):} 
For every $0.618 \approx \varphi-1 < \beta \leq 1$, there exists an instance with two subadditive valuations that admits no randomized allocation that is simultaneously ex-ante $\alpha$-EF and ex-post $\beta$-EFX, {for any} $\alpha > \frac{\beta+1}{\beta^2+2\beta}$ . 
}

\vspace{0.1in}
In particular, our impossibility result implies the following upper bounds:
\begin{itemize}
    \item There is no random allocation that is ex-ante $\alpha$-EF and ex-post EFX for $\alpha>\frac{2}{3}$.
    \item There is no random allocation that is ex-ante EF and ex-post $\beta$-EFX for $\beta > \varphi-1$.
\end{itemize}

Finally, in \Cref{sec:two-agents} we show that the first of these bounds is tight for two agents, namely, we devise an algorithm that gives an ex-ante $\frac{2}{3}$-EF and ex-post EFX allocation for every instance with two subadditive valuations (\Cref{prop:subadditive-2-agents-efx}). 
In addition, for every such instance, we devise an algorithm that gives an ex-ante EF and ex-post $\frac{1}{2}$-EFX allocation (\Cref{prop:subadditive-2-agents-ef}).

Our results are summarized in \Cref{fig:tradeoffs}.

\begin{figure}[h!]
\centering
    \input{figures/tradeoffs}
    \caption{Trade-offs between {ex-ante} $\alpha$-EF and {ex-post} $\beta$-EFX, for two (left) and $n$ (right) subadditive agents. 
    The $x$-axis (resp., $y$-axis) represents ex-ante $\alpha$-EF (resp., ex-post $\beta$-EFX) guarantees.
    Dotted (resp., hatched) areas represent
    our existence (resp., impossibility) results.}
    \label{fig:tradeoffs}
\end{figure}
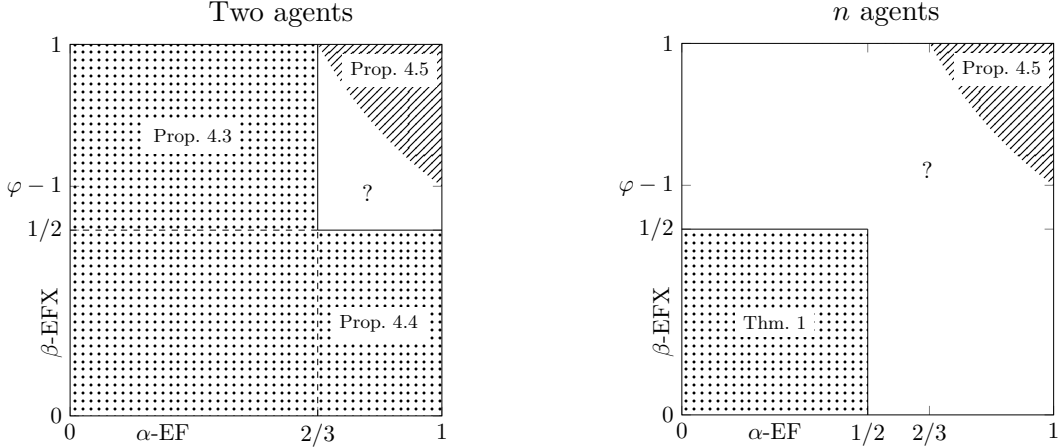

\subsection{Our Techniques}

In this section, we present an overview of the construction and the analysis of our main algorithm, {which proposes a way to randomize the widely-used envy cycles procedure~\cite{LMMS04} in order to ensure ex-ante fairness.}

The standard deterministic envy cycles procedure involves two phases. In the first phase, the algorithm assigns one item to each agent using a matching algorithm that satisfies a natural property called weak separation (\Cref{def:weak_separation}), {which states that every agent prefers her item to any unassigned item.}
In the second phase, {the algorithm allocates the remaining items}. This is done by either (1) allocating one of the remaining items to an {\em unenvied} agent, or (2) shifting the bundles along an {\em envy cycle}, which is a cycle where each agent prefers the next agent's bundle to her own. This procedure is explained in detail in \Cref{sec:deterministic_envy_cycles}.

Crucially, there are many ways to implement the deterministic envy cycles procedure that yield the same ex-post guarantees. For instance, in the second phase,
one can arbitrarily decide which unenvied agent to select for operation (1) or which envy cycle to choose for operation (2). However, for \emph{any} deterministic choices for these operations, this algorithm gives no ex-ante EF guarantees. Consequently, achieving the desired ex-ante fairness requires randomization of these choices. In what follows, we describe the randomization of each phase of the algorithm separately.

\vspace{0.25cm}
\textbf{Randomizing the First Phase.}
One possible approach to obtain ex-ante fairness guarantees is to apply some randomness to a given deterministic algorithm.
For example, a common algorithm for the first phase is
{serial dictatorship}, where 
agents are sorted according to some order, and every agent, upon her turn, chooses her most preferred item among the remaining ones.
One can easily verify that {deterministic}
{serial dictatorship} gives no ex-ante fairness guarantees. 
A natural approach to obtain {some} ex-ante fairness guarantees is to randomize the order in which agents choose items. 
This is called random serial dictatorship (RSD). 
Unfortunately, {choosing} a uniformly random order over the agents does not give ex-ante EF \cite{BM01}. In fact, we {improve the upper bound on the ex-ante fairness of {RSD}, by showing ({\Cref{prop:upper-bound-rrr}}) that} using a uniformly random order does not give better than ex-ante $\frac{1}{\sqrt{2}}$-envy-freeness.

Therefore,
to construct the matching in the first phase, we use the
{\em probabilistic serial} procedure of~\citet{BM01}.
Recently, \citet{FSV20} and~\citet{A20} utilized this procedure to achieve 
{ex-ante EF together with ex-post EF1 for additive valuations.}
{They apply the probabilistic serial procedure for the entire allocation process. 
Unfortunately, this method does not extend to subadditive valuations, and even for additive valuations, it does not guarantee $\beta$-EFX for any $\beta > 0$ (see \Cref{sec:ps_explained}).
Instead, we apply the probabilistic serial method for a single time unit, followed by randomized envy cycles.} 
{See \Cref{sec:matching_distribution} for more details.}

\vspace{0.25cm}
\textbf{Randomizing the Second Phase.} It may seem reasonable to hope that selecting the envy cycles arbitrarily during the second phase of the algorithm results in a fair outcome.
However, even though an envy cycle elimination step does not modify any of the bundles, it might shuffle them among the agents in a way that benefits some agents more than others. This is illustrated in \Cref{ex:detcycles}.
{A natural second attempt would be}
to choose the envy cycle uniformly at random. 
However, one can show that this too does not provide any ex-ante EF guarantees.

We propose a novel way to sample the envy cycles in a way that provides a probability distribution {that gives ex-ante EF guarantees}.
To construct the desired distribution, we present an algorithm inspired by the decomposition of irreducible Markov chains into circuit processes by~\citet{M81}.
This algorithm applies only to strongly connected graphs, and to address this challenge,
we first select a strongly connected component of the envy graph that has no incoming edges (which always exists), and run the algorithm on the induced subgraph.
Details regarding the randomization of the second phase are given in \Cref{sec:envy_cycles_distribution}.

\vspace{0.25cm}
Following the explanation of the main algorithm's design, we now focus on the components of the analysis of its ex-ante fairness guarantees.
{We first identify two key properties of the algorithm}.

{First, the random allocation obtained in the first phase satisfies a property that we term {\em strong separation} (see \Cref{lem:ps_strong_sep}).
This property states that every agent prefers her bundle in any allocation in the support of the random allocation to any item that may be unallocated in any (possibly different) allocation in the support.}
Second, we show that after any step of the algorithm, the value of every agent $i$ for her own bundle {\em stochastically dominates} the value of $i$ for (a specific part of) any other agent's bundle (see \Cref{lem:sd_ineq}).

{To establish our main result, we first use the stochastic dominance property (\Cref{lem:sd_ineq}) to show that {\em two copies} of $X_i$ can ``cover'' a specific subset of $X_j$, in a way that is captured by the definition of {\em stochastic coverage} that we introduce in \Cref{sec:proof}. The proof is by induction on the execution tree, i.e., the tree representing all possible random choices that the algorithm makes on a given input. Moreover, using strong separation (\Cref{lem:ps_strong_sep}), we extend this property to show that the two copies of $X_i$ can also cover the remaining part of $X_j$, from which our main result follows easily. The full analysis of the main algorithm is given in \Cref{sec:proof}.}

\subsection{Related Work}\label{sec:related-work}
While the fair division problem has been central to human society since antiquity, its formal study began with the work of Banach, Knaster and Steinhaus~\cite{Ste48}, who analyzed the \emph{cake-cutting} problem: \emph{how can a heterogeneous cake be divided fairly amongst agents?} With $n=2$ agents, the folklore \emph{cut-and-choose} method results in a \emph{proportional} allocation of the cake, i.e., one in which each agent is allocated a piece of value at least $\frac{1}{n}$ of their value {($\frac{1}{2}$ in the case of $n=2$)} for the whole cake. This result was extended to $n>2$ agents with the \emph{last diminisher} procedure~\cite{Ste48}. Subsequent decades witnessed the emergence of \emph{envy-freeness} (\citet{Fol67,V74}) as the main criterion for fairness in economics. An allocation is envy-free if each agent prefers her own bundle to the bundle of any other agent. For the settings with divisible items, which generalize the cake-cutting problem, envy-free allocations are known to exist under mild assumptions~(\citet{Str80}), and many algorithms have been devised for their computation~\cite{BT95,AM16}.

\vspace{0.25cm}
\textbf{The Envy Cycles Procedure.} For the complementary setting with {\em indivisible} items, both proportionality and envy-freeness are impossible to achieve: consider the simple instance with two agents and one item, where any allocation leaves one agent envying the other. A breakthrough development in the area occurred almost two decades ago, when \citet{LMMS04} introduced the now-ubiquitous \emph{envy cycles procedure} for the fair division of indivisible items. This led to the development of relaxed fairness notions such as the EF1 guarantee, formally defined by \citet{B11}. The envy cycles procedure computes, in polynomial time, an EF1 allocation even for the very general class of monotone valuations. This work subsequently spawned a large body of research in the fair division of indivisible items, and the envy cycles procedure (along with its variations) {has} since been used in a wide variety of studies in fair division, including in algorithms for finding approximate MMS allocations~\cite{BK17}, EF1 allocations for chore division~\cite{BSV21}, approximate EFX allocations~\cite{PR18}, and partial EFX allocations (``EFX with charity'')~\cite{CKMS20}, among several others.

\vspace{0.25cm}
\textbf{The EFX and MMS Fairness Notions.} The EFX guarantee was introduced by~\citet{CKMPSW16}. Despite its apparent similarity to EF1, the existence of EFX allocations remains a notoriously hard open problem even for additive agents. Complete EFX allocations are known to exist for instances with at most three additive agents~(\citet{CGM20}), and in the case where the agents have identical valuation functions~(\citet{PR18}). Beyond these results, only improvements for special cases {or for various relaxations} have been found~\cite{PR18,M21,ACGMM22,GHNV23}. Because of the limited progress on the existence of complete EFX allocations, both \emph{partial}- and \emph{approximate}-EFX allocations have also been studied. Numerous works have studied the existence of partial allocations that provide some EFX guarantee while allocating almost all of the items~\cite{CKMS20,CGMMM21,BCFF22,ACGMM22,FMP23}. For approximately-EFX fair division, \citet{PR18} showed that a $\frac{1}{2}$-EFX allocation exists even for subadditive valuations. For additive valuations, \citet{AMN20} improved this approximation factor to $\varphi-1 \approx 0.618$. 

Another related fairness notion for the indivisible setting, defined by~\citet{B11}, is the maximin share (MMS) guarantee. It is known that allocations that give each agent its MMS value do not exist even for additive agents~(\citet{PW14,FST21}), but there are allocations that give each agent a constant factor of the MMS value, and these allocations can be efficiently computed even for submodular and XOS valuations (see, e.g.,~\cite{BK17,GHSSY18}).

\vspace{0.25cm}
\textbf{Best-of-Both-Worlds Fairness.}
A recent line of research, titled best-of-both-worlds fairness, aims to achieve envy-freeness for indivisible items via randomization. The leading question is whether it is possible to simultaneously achieve some ex-ante fairness guarantee while ensuring ex-post guarantees for every realized outcome of the random process. {\citet{AAGW15} first studied this problem in a food bank setting, where they gave an algorithm with ex-ante and ex-post fairness guarantees for the {special case} of additive valuations with \emph{binary (0/1)} marginals.} For additive valuations, \citet{FSV20} gave a randomized polynomial-time algorithm that outputs an EF1 allocation, while being envy-free ex-ante. Ensuing work by~\citet{A20} showed that there exists a similar randomized algorithm that additionally implements the well-known \emph{Probabilistic Serial} fractional outcome described by~\citet{BM01}. This randomized algorithm also preserves a weak notion of efficiency. \citet{BEF22} also study the case of additive valuations and find a distribution over ex-post {proportional up to one item} and $\frac{1}{2}$-MMS allocations that is ex-ante proportional, i.e., the expected value of each agent's bundle is at least a $\frac{1}{n}$-fraction of her value for the set of all items.

Best-of-both-worlds results have also been analyzed for other settings, including the case where the agents have binary marginals, and the case of additive valuations with \textit{arbitrary entitlements}. For the {special} case of additive valuations with binary marginals, the algorithm of~\citet{A20} {improves upon the guarantees of~\citet{AAGW15} and} is group-strategyproof, ex-ante fractionally-PO and envy-free, and ex-post fractionally PO and EF1. In a similar vein, for additive valuations with binary marginals, \citet{HPPS20} independently showed that there is a distribution over ex-post Nash-welfare-maximizing allocations that also ex-ante maximizes the fractional Nash welfare, implying the same fairness guarantees as \citet{A20} for this setting. For \emph{matroid rank} valuations, \citet{BEF21} present a randomized truthful mechanism that is ex-ante envy-free and ex-post Lorenz dominating (and thus ex-post Nash-welfare-maximizing and EFX for this specific class). For the case of agents with arbitrary entitlements, both \citet{HSV23} and \citet{AGM23} show that weighted ex-ante envy-freeness along with ex-post weighted envy-freeness up to 1 item is impossible to achieve. Both \citet{HSV23} and \citet{AGM23} then give polynomial time algorithms that instead achieve some weaker fairness properties. Additionally, \citet{HSV23} show that ex-ante EF and ex-post EF1 can be achieved for multi-demand valuations.

\section{Preliminaries}

An instance of the resource allocation problem consists of a set $[n] = \{1, \ldots, n\}$ of $n$ agents, a set $[m] = \{1, \ldots, m\}$ of $m$ indivisible items, and a valuation profile $(v_i)_{i \in [n]}$. The valuation function $v_i : 2^{[m]} \to \mathbb{R}^{\geq 0}$ of agent $i$ gives a non-negative real value $v_i(S)$ for every bundle of items $S \subseteq [m]$. {We assume that {all valuation functions}
are monotone, i.e., $v_i$ satisfies $v_i(S) \leq v_i(T)$ for any $S \subseteq T \subseteq [m]$ {and every agent $i$}.} The valuation functions we consider in this work are either
\begin{itemize}
    \item additive, i.e., it holds that $v_i(S) = \sum_{x \in S} v_{ix}$ for some $v_{i1}, \ldots, v_{im} \geq 0$, or
    \item subadditive, i.e., it holds that $v_i(S \cup T) \leq v_i(S) + v_i(T)$ for any two $S, T \subseteq [m]$.
\end{itemize}
We use the following standard notation for singleton sets: for any $g \in [m]$ and $S \subseteq [m]$, we write $v_i(g) = v_i(\{g\})$ and $S+g = S \cup \{g \}$ and $S-g = S \setminus \{ g \}$.

A {\em deterministic} partial allocation $X = (X_1, \ldots, X_n)$ of items to agents is a partition of the items into $n$ bundles where all bundles are disjoint, i.e., it holds that $X_i \cap X_j = \emptyset$ for all $i \neq j$. We say that $X$ is {\em complete} if no items are unallocated, i.e., it holds that $\bigcup_{i \in [n]} X_i = [m]$.

We say that a deterministic allocation $X$ is $\alpha$-EF1 (for some $\alpha \in [0,1]$) if for all agents $i$ and $j$, it holds that $v_i(X_i) \geq \alpha \cdot v_i(X_j-g)$ for {\em some} item $g \in X_j$. We say that an allocation $X$ is $\alpha$-EFX (for some $\alpha \in [0,1]$) if it holds that $v_i(X_i) \geq \alpha \cdot v_i(X_j - g)$ for {\em all} items $g \in X_j$.
An allocation is EF1 (resp., EFX) if it is $1$-EF1 (resp., $1$-EFX).

\vspace{0.25cm}
\textbf{Random allocations, ex-post and ex-ante fairness.} 
A {\em random} allocation $X = (X_1, \ldots, X_n)$ is a random variable $X$ whose every random outcome is a deterministic allocation. Here, $X_1, \ldots, X_n$ are the associated random variables whose every random outcome is a bundle of items.
A random allocation $X$ is {\em ex-post} $\alpha$-EF1 {(resp., $\alpha$-EFX)} if every random outcome of $X$ (that happens with non-zero probability) is an $\alpha$-EF1 {(resp., $\alpha$-EFX)} deterministic allocation.
We say that $X$ is {\em ex-ante} $\alpha$-EF (for some $\alpha \in [0,1]$) if $\mathbb{E}[v_i(X_i)] \geq  \alpha \cdot \mathbb{E}[v_i(X_j)]$ for all $i,j$. An allocation is EF if it is $1$-EF.

\vspace{0.25cm}
\textbf{Envy graph {and strongly connected components}.} For every deterministic partial allocation $X = (X_1, \ldots, X_n)$, we define the envy graph $G_X = ([n],E_X)$ where the set of nodes is the set of agents $[n]$ and the set of directed edges $E_X = \{ (i,j) \in [n]^2 : v_i(X_j) > v_i(X_i) \}$ includes an edges from agent $i$ to agent $j$ if $i$ envies $j$, i.e., agent $i$ prefers agent $j$'s bundle over her own bundle.

A {directed} graph is {\em strongly connected} if there exists a directed path from any node to any other node. A {\em strongly connected component} of a directed graph $G$ is {any maximal (in the sense of inclusion) subgraph of $G$} that is strongly connected.

\vspace{0.25cm}
\textbf{Stochastic dominance.} Let $X$ and $Y$ be any two non-negative {real} random variables.
We say that $X$ {\em stochastically dominates} $Y$, and we write $X \SD Y$, if it holds that $\mathbb{P}[X \geq t] \geq \mathbb{P}[Y \geq t]$ for every $t \geq 0$.

\section[Our Main Result: Ex-Ante 1/2-EF, Ex-Post 1/2-EFX, Ex-Post EF1]{Our Main Result: Ex-Ante $\frac{1}{2}$-EF, Ex-Post $\frac{1}{2}$-EFX, Ex-Post EF1}

In this section, we describe our main result: the existence of a distribution over $\frac{1}{2}$-EFX and EF1 allocations that is ex-ante $\frac{1}{2}$-{EF} for every instance with subadditive valuations, and a randomized algorithm that outputs a random allocation with this distribution in polynomial time. This result is presented in \Cref{thm:main_theorem}, which we restate below.

\begin{theorem}
\label{thm:main_theorem}
    For every instance with subadditive valuations, \Cref{alg:fair_envy_cycles} outputs a random allocation that is ex-ante $\frac{1}{2}$-EF, ex-post $\frac{1}{2}$-EFX, and ex-post EF1 in polynomial time.
\end{theorem}

We may assume without loss of generality that there are more than $n$ items. Indeed, if there are $m\leq n$ items, then we may add $n-m$ dummy items with value $0$ for all agents, and choose a random allocation that assigns one item per agent. The resulting allocation is EFX (since each agent is assigned at most one item) and ex-ante EF (by symmetry).

\subsection{A Deterministic Envy Cycles Procedure}
\label{sec:deterministic_envy_cycles}

{In this section we describe a deterministic algorithm} for constructing $\frac{1}{2}$-EFX and EF1 allocations {(but, with no EF guarantees)} when the agents have subadditive valuations. 
As we will show in the following subsections, our main algorithm is a carefully randomized implementation of this deterministic algorithm.

The existence of $\frac{1}{2}$-EFX and EF1 allocations for subadditive agents was first shown by \citet{PR18}.
Their proof is constructive and involves the use of an algorithm that is based on the envy cycles procedure. {We have modified the presentation of the algorithm to suit our analysis, and we present it here as a deterministic two-phase algorithm that runs in polynomial time, {see \Cref{alg:det_envy_cycles}} in \Cref{sec:det_envy_cycles}}.

In the first phase, the algorithm constructs a \emph{matching} between the agents and the items, in which exactly one item is assigned to each agent. The matching must satisfy the following property, called \emph{weak separation}: each agent prefers the item she receives in the matching to any of the remaining unassigned items. The formal definition of this property {was explicitly given in \citet{FMP23}, but was implicitly used in previous studies on EFX allocations \cite{AMN20}}.

\begin{definition}[Weak separation] \label{def:weak_separation}
    A partial allocation $X=(X_1,\ldots,X_n)$ with the set $U = \items \setminus \bigcup_{i \in \agents} X_i$ of unallocated items satisfies \emph{weak separation} if for every agent $i\in[n]$ and every unallocated item $u\in U$ it holds that $v_i(X_i) \geq v_i(u)$.
\end{definition}

In the second phase, the algorithm performs the widely-known envy cycles procedure of~\citet{LMMS04} to allocate the remaining unassigned items, {as follows}. 
Suppose that the algorithm has already constructed a partial allocation $X=(X_1,\ldots,X_n)$. 
Then, in the next step, 
if there exists an {\em unenvied} agent $j$ (i.e., an agent $j$ such that $v_i(X_i) \geq v_i(X_j)$  for all $i\neq j$), then the algorithm selects an arbitrary unallocated item $g$ (i.e., $g$ is not in the set $\bigcup_{i \in [n]} X_i$), and assigns {item $g$} to agent $j$.
If there is no unenvied agent, the algorithm finds a directed cycle $C=(u_1,u_2,\ldots,u_k)$ in the envy graph, where agent $u_1$ envies agent $u_2$, agent $u_2$ envies agent $u_3$, \ldots, and agent $u_k$ envies agent $u_1$.  
Observe that such a cycle always exists; indeed, if there is no unenvied agent, every node in the envy graph has an incoming edge. 
The algorithm then performs a \textit{cycle elimination} step, in which the fixed bundles are redistributed along the cycle {(namely, every agent $u_i$ receives the bundle $X_{u_{i+1}}$, where $u_{k+1} = u_1$)}.

{We first show the following monotonicity lemma, whose proof is deferred to \Cref{sec:det_envy_cycles}.}

\begin{restatable}{lemma}{lembundlemonotonicity}\label{lem:bundle_monotonicity}
    For any agent $i$, the value of $i$ for her own bundle weakly increases at each step of the envy cycles procedure.
\end{restatable}

{Next, we show that the {final} allocation of this algorithm is $\frac{1}{2}$-EFX and EF1.
This is shown by \citet{PR18} for their algorithm, but as our algorithm is presented in a slightly different way, we provide the proof in \Cref{sec:det_envy_cycles} to ensure completeness.
}

\begin{restatable}[\cite{PR18}]{lemma}{lemexpost}\label{lem:ex_post}
    The allocation returned by \Cref{alg:det_envy_cycles} is $\frac{1}{2}$-EFX and EF1.
\end{restatable}

We also show in \Cref{sec:det_envy_cycles} that our algorithm runs in polynomial time.

\begin{restatable}{lemma}{lemcyclespolytime}\label{lem:cycles-polytime}
   \Cref{alg:det_envy_cycles} terminates after polynomially many steps. 
\end{restatable}

\subsection{A Distribution over Matchings}\label{sec:matching_distribution}

For the first phase of the algorithm, we follow the approach of \citet{BM01}, who study the allocation of $n$ items among $n$ agents, where every agent receives exactly one item.
Their algorithm applies Probabilistic Serial, followed by Birkhoff-von Neumann rounding, described below.

In Probabilistic Serial, the allocation is constructed via the {\em simultaneous eating} procedure, where every agent ``consumes'' her favorite item (with ties between multiple items of the same value broken arbitrarily) at a constant rate of one item per one unit of time, i.e., all agents have the same eating rate. After any item is fully consumed (possibly by multiple agents), each agent independently and instantaneously switches to her next-best item (again, with ties broken arbitrarily) and continues eating at the same constant rate.

\citet{BM01} study the setting where $n$ items are allocated to $n$ agents, and every agent receives exactly one item. In our case, however, the number of items $m$ might be greater than the number of agents $n$, and we apply this procedure for exactly one time unit (i.e., at the point where each agent has consumed a total fractional mass of exactly one item). This is in contrast to the generalization of \citet{FSV20} and \citet{A20}, who repeatedly apply probabilistic serial until all $m$ items are allocated. We refer the reader to \cref{sec:ps_explained} for a more detailed explanation of Probabilistic Serial and the algorithms of \citet{FSV20} and \citet{A20}.

The output of the algorithm is a fractional allocation of items to agents (see \Cref{fig:ps_eating}(a)). 
It is represented by a matrix $Z_{ij} \in [0,1]$ for $1 \leq i \leq n$ and $1 \leq j \leq m$, satisfying $\sum_{j \in [m]} Z_{ij} = 1$ for every agent $i \in [n]$ and $\sum_{i \in [n]} Z_{ij} \leq 1$ for every item $j \in [m]$.

We then apply a classic theorem of Birkhoff~\cite{B46} and von Neumann~\cite{N53} (which we restate here in a form that is relevant to our algorithm) to decompose the fractional allocation returned by the eating procedure into a distribution over integral allocations.

\begin{theorem}[Birkhoff-von-Neumann] \label{thm:bvn}
    Let $Z$ be the fractional allocation returned after one time step of Probabilistic Serial. There is a strongly-polynomial-time algorithm that computes binary matrices $X^1,\ldots,X^q$ (where $X^k_{ij}\in\{0,1\}$ for all $i,j,k$) and probabilities $p^1,\ldots,p^q \in (0,1]$ (where $\sum_{k\in[q]}p^k = 1$) such that $Z = \sum_{k\in[q]}p^kX^k$, i.e., the distribution $((X^k,p^k))_{k\in[q]}$ over integral allocations is a weighted decomposition of the fractional allocation $Z$.
\end{theorem}

The above theorem immediately implies the following corollary.

\begin{corollary} \label{cor:bvn}
For the fractional allocation $Z$ returned by one time step of Probabilistic Serial, and its decomposition $((X^k,p^k))_{k\in[q]}$ obtained in \Cref{thm:bvn}, we have
    \begin{enumerate}[\label=(i)]
        \item for every agent $i\in[n]$ and $k\in[q]$, it holds that $\sum_{j\in[m]} X^k_{ij} = 1$, i.e., each agent is assigned exactly one item in every integral allocation $X^k$ in the support of the decomposition.
        \item for every item $j\in[m]$ for which $\sum_{i\in[n]}Z_{ij} = 0$, it holds that $\sum_{i \in [n]}X^k_{ij} = 0$, i.e., the item $j$ is unallocated in every integral allocation $X^k$ in the support of the decomposition.
        \item for every item $j\in[m]$ for which $\sum_{i\in[n]}Z_{ij} = 1$, it holds that $\sum_{i \in [n]}X^k_{ij} = 1$, i.e., the item $j$ is allocated in every integral allocation $X^k$ in the support of the decomposition.\label{cor:bvn_allocated}
        \end{enumerate}
\end{corollary}

Sampling an integral allocation $X = (X_1, \ldots, X_n)$ from the distribution given by \Cref{thm:bvn} has several advantages. 
First, it has strong ex-ante properties, captured by the following lemma.

\begin{restatable}[\cite{BM01}]{lemma}{lempsprefix}\label{lem:ps_prefix}
    {For any allocation $X$ obtained by \Cref{thm:bvn}}, it holds that $v_i(X_i) \SD v_i(X_j)$
    for any two agents $i$ and $j$.
\end{restatable}

\begin{proof}
    We show that $\mathbb{P}[v_i(X_i) \geq t] \geq \mathbb{P}[v_i(X_j) \geq t]$ for any $t \geq 0$. Let $\mathcal{A} = \{ a \in \items : v_i(a) \geq t\}$. Note that $\mathbb{P}[v_i(X_i) \geq t] = \sum_{a\in \mathcal{A}} Z_{ia}$ and $\mathbb{P}[v_i(X_j) \geq t] = \sum_{a\in\mathcal{A}} Z_{ja}$. Suppose for the purpose of contradiction that $t_1 = \sum_{a \in \mathcal{A}} Z_{ia} < \sum_{a \in \mathcal{A}} Z_{ja} = t_2$. This means that after $t_1$ units of time, agent $i$ started consuming an item of value less than $t$. However, between time $t_1$ and $t_2$, agent $j$ was still consuming some items in $\mathcal{A}$ which means that at least one of them was still not fully consumed during that time. This contradicts the assumption that $i$ always switches to eat the next most valuable item that is not fully consumed yet. 
\end{proof}

Moreover, it is easy to see that any integral allocations in the support of the distribution given by \Cref{thm:bvn} is weakly separated (\Cref{def:weak_separation}). 
Hence, the ex-post guarantees given in \Cref{sec:deterministic_envy_cycles} are always satisfied.

However, to prove the desired ex-ante guarantees given in \Cref{sec:proof}, weak separation is not enough, and we use the following stronger notion.

\begin{definition}[Strong separation]
    A \emph{random} partial allocation $X=(X_1,\ldots,X_n)$ with the corresponding \emph{random} collection $U = \items \setminus \bigcup_{i \in \agents} X_i$ of unallocated items satisfies \emph{strong separation} if it holds that for any allocation $Y$ in the support of $X$, we have $v_i(Y_i) \geq v_i(w)$ for every agent $i\in[n]$ and every item $w$ such that $\mathbb{P}[w\in U] > 0$.
\end{definition}

Clearly, any strongly separated allocation is also weakly separated.
We now show that the outcome of the first phase is strongly separated.

\begin{lemma} \label{lem:ps_strong_sep}
    The (random) allocation $X=(X_1, \ldots X_n)$ is strongly separated.
\end{lemma}

\begin{proof}
    Fix any agent $i$. Let $a \in [m]$ be any item that agent $i$ ate a positive fraction of, i.e., it holds that $Z_{ia} = \mathbb{P}[X_i = \{a\}] > 0$. Let $b \in [m]$ be any item that was not fully consumed during the eating process, i.e., it holds that $1 - \sum_{k \in [n]} Z_{kb} = \mathbb{P}[b \in U] > 0$. Then, it is the case that $v_i(a) \geq v_i(b)$. Otherwise, agent $i$ would have started eating $b$ before eating any fraction of $a$.
\end{proof}

\subsection{A Distribution over Envy Cycles}\label{sec:envy_cycles_distribution}

Having described the implementation of the first phase of the algorithm, let us now discuss the second phase. Suppose that the algorithm already selected a partial allocation $(X_1, \ldots, X_n)$ and at the current step, the algorithm must either (1) assign an unallocated item to an unenvied agent if there is such an agent, or (2) eliminate an envy cycle by reallocating the bundles along the cycle if there is such a cycle. 

First, it is important to note that the ex-post guarantees of the deterministic algorithm given by \Cref{lem:ex_post} do not depend on the specific choices of the algorithm. The algorithm can follow an arbitrary rule to decide (a) whether to perform operation (1) or (2) if at some point both of them can be executed, (b) which unenvied agent to give a new item to, (c) which unallocated item to assign to an unenvied agent, and (d) which envy cycle to eliminate.

{However, when it comes to ex-ante guarantees, arbitrary choices regarding the above dimensions can lead to an arbitrarily high ex-ante envy.
This is demonstrated in the following example.}

\begin{example}\label{ex:detcycles}
    {Consider an instance with $n$ agents having additive valuations over $2n+1$ items, as described in the table below.}
    \begin{center}
        \begin{tabular}{||c | c c c |ccc| c c c||} 
         \hline
         & apple & banana & celery & \multicolumn{3}{c|}{durians $\times$ ($n-3$)} & \multicolumn{3}{c||}{eggplants $\times$ ($n+1$)} \\ [0.5ex] 
         \hline\hline
         $v_1$ & $10$ & $3\varepsilon$ & $2\varepsilon$ & \;\;$0$ & \quad$\dots$\quad & $0$ & \;\;$2\varepsilon$ & \quad$\dots$\quad & $2\varepsilon$\\
         $v_2$ & $10$ & $2\varepsilon$ & $3\varepsilon$ & \;\;$0$ & \quad$\dots$\quad & $0$ & \;\;$2\varepsilon$ & \quad$\dots$\quad & $2\varepsilon$\\
         \hline
         $v_3$ & $10$ & $8$ & $0$ & \;\;$9$ & \quad$\dots$\quad & $9$ & \;\;$3$ & \quad$\dots$\quad & $3$\\
         $v_4$ & $10$ & $8$ & $0$ & \;\;$9$ & \quad$\dots$\quad & $9$ & \;\;$3$ & \quad$\dots$\quad & $3$\\
         $\vdots$ & $\vdots$ & $\vdots$ & $\vdots$  &\;\;$\vdots$  & \quad$\ddots$\quad & $\vdots$ & \;\;$\vdots$  & \quad$\ddots$\quad & $\vdots$ \\
         $v_n$ & $10$ & $8$ & $0$ & \;\;$9$ & \quad$\dots$\quad & $9$ & \;\;$3$ & \quad$\dots$\quad & $3$\\
         \hline
        \end{tabular}
    \end{center}
\end{example}

Observe that in this example, the apple is the most-preferred item for every agent. Suppose that in the first phase of the algorithm,
the following allocation is chosen: agent~$i$ (for some $3 \leq i \leq n$) gets the apple, agent~$1$ gets the banana, agent~$2$ gets celery, and each of the remaining agents~$3, \ldots, i-1, i+1, \ldots, n$ gets a durian. This is, for instance, exactly the allocation that is chosen with high probability (i.e., probability $1- o(1)$) by both the probabilistic serial method and the serial dictatorship algorithm with a uniformly random permutation of the agents. 

Furthermore, suppose that in the first step of the second phase, agent~$1$ (who is unenvied at the moment) gets an eggplant (which is unallocated at the moment). After this operation, there is an envy cycle: agent~$1$ (who has the banana and an eggplant) envies agent~$i$ (who has the apple), who in turn envies agent~$1$. However, if the algorithm eliminates this cycle, then agent~$2$ (who has the celery) envies agent~$1$ (who now has the apple) by an arbitrarily large factor. Moreover, for sufficiently small $\varepsilon$, once agent~$1$ receives the apple she will not be included in any envy cycles in the subsequent steps, so the apple will stay with her until the {termination of the algorithm.}

\Cref{ex:detcycles} shows that during a cycle elimination step, selecting the envy cycle in an arbitrary way does not give ex-ante $\alpha$-EF for any $\alpha > 0$.
To overcome this problem, we employ two ideas.
First, we only eliminate envy cycles if there are no unenvied agents.
{In \Cref{ex:detcycles}}, this means that we do not immediately allow agent~1 to trade for the apple, because agent~2 remains unenvied.

The second {(and key)} idea is the following. Suppose that there are two agents~$i_1$ and~$i_2$ who both envy some bundle $X_j$. We want to avoid selecting an envy cycle containing the arc $(i_1, j)$ where $i_1$ gets $X_j$, unless, with sufficiently high probability, we also select some other envy cycle containing the arc $(i_2, j)$, where $i_2$ gets $X_j$. Otherwise, agent~$i_2$ might have high ex-ante envy towards agent~$i_1$. The corresponding intuition in \Cref{ex:detcycles} is that we only allow agent~$1$ to trade for the apple if agent~$2$ also has a sufficiently high probability of getting the apple. 

Now, if the envy graph is strongly connected, then applying the following key lemma to the envy graph gives a ``fair'' distribution over the envy cycles. 
{The proof of this lemma is inspired by the work of~\citet{M81} on circuit processes. \Cref{fig:cycles} provides an illustration for this proof.}

\begin{lemma}{\bf (Key Lemma)}
\label{lem:cycle_distribution}
    Let $G = (V, E)$ be a strongly connected directed graph with $|V| \geq 2$. There exists a probability distribution $((c_t, p_t))_{t\in[r]}$ over the set of simple cycles $c_t$ in $G$ such that for all $j \in V$ and $(i_1, j), (i_2, j) \in E$ it holds that $\sum_{t:(i_1, j) \in c_t} p_t = \sum_{t:(i_2, j ) \in c_t} p_t$, i.e., the total probability {of} all cycles containing the edge $(i_1, j)$ is equal to total probability {of} all cycles containing the edge $(i_2,j)$. {Moreover, such a distribution can be computed in polynomial time.}
\end{lemma}
\begin{proof}
    Consider a Markov chain with the set of states $V$ and the transition probabilities given by
    \begin{align*}
        p(j,i)=
            \begin{cases}
            1/\text{in-deg}_j & \text{if } (i,j) \in E \\
            0 & \text{otherwise}
            \end{cases}
    \end{align*}
    The probabilities $p$ are well-defined because $\sum_{i \in V} p(j,i) = 1$ for all $j \in V$.
    Since $G$ is strongly connected, this Markov chain is irreducible and therefore has a stationary distribution $\pi = (\pi_i)_{i \in V}$. Note that the stationary distribution is unique~\cite[Section 1.7]{N98}, and we can therefore compute the distribution in polynomial time with a direct method such as Gaussian elimination.

    Define the flow $w(j,i) = \pi(j) \cdot p(j,i)$ and observe that $w$ is a balanced flow, i.e., the flow that goes into $i$ is equal to the flow that goes out of $i$. This is true because 
    $$\sum_{i \in V} w(j,i) = \pi(j) = \sum_{i \in V} \pi(i) \cdot p(i,j) = \sum_{i \in V} w(i,j),$$
    where the second equality follows from the properties of the stationary distribution $\pi$.

    Next, we show that there is a probability distribution $((c_t, p_t))_{t\in[k]}$ over the set of simple cycles $c_t$ in $G$ such that for all $(i,j) \in E$ it holds that $\sum_{t:(i,j) \in c_t} p_t = w(j,i) \cdot C$ for some constant $C > 0$. This implies that for any two edges $(i_1, j), (i_2,j) \in E$ it holds that \[ \sum_{t:(i_1,j) \in c_t} p_t \cdot C^{-1} = w(j,i_1) = \pi(j) \cdot p(j,i_1) = \pi(j) / \text{in-deg}_j = w(j,i_2) = \sum_{t:(i_2,j) \in c_t} p_t \cdot C^{-1}.\] 
    
    We construct our collection of cycles in the following manner. We start with an empty collection of cycles. In each step $t$, we find a cycle $((i_1, i_2), (i_2, i_3), \ldots, (i_{k-1},i_k), (i_k, i_1))$ on which every arc has positive flow (this can be done in polynomial time via, e.g., depth-first search), and then find the minimum flow $w_t$ over any arc in this cycle. We add the pair $(c_t,p_t)$ to our collection where $c_t = ((i_1,i_k),(i_k,i_{k-1}),\ldots,(i_3,i_2),(i_2,i_1))$ and $p_t=w_t$, and then reduce the flow on the arcs $(i_1,i_2), (i_2,i_3), \ldots, (i_{k-1},i_k)$ by $w_t$. Note that the flow is still balanced after this operation. We repeat this process until {there no longer exists} a cycle with positive flow on every arc. 
    
    Since in each step of the above process, the flow on at least one arc is reduced to 0, this process {terminates} after a polynomial number of steps. Additionally, since $w$ is balanced at every step, this process only terminates when the flow is zero everywhere and we thus obtain $\sum_{t : (i,j) \in c_t} p_t = w(j,i)$. Finally, we normalize the probabilities over the cycles in our collection so that they sum to 1.
\end{proof}

\begin{figure}
\centering
\input{figures/cycles.tex}
\caption{Illustration of the proof of \Cref{lem:cycle_distribution} for a strongly connected graph. The graph in (i) induces the stationary distribution in the Markov chain presented in (ii), which can be used to compute the cycle distribution presented in (iii). The multiplicity of each arrow in (ii) is proportional to its flow, as defined by $w$ in the proof, i.e., each line represents a flow of $1/13$.}
\label{fig:cycles}
\end{figure}
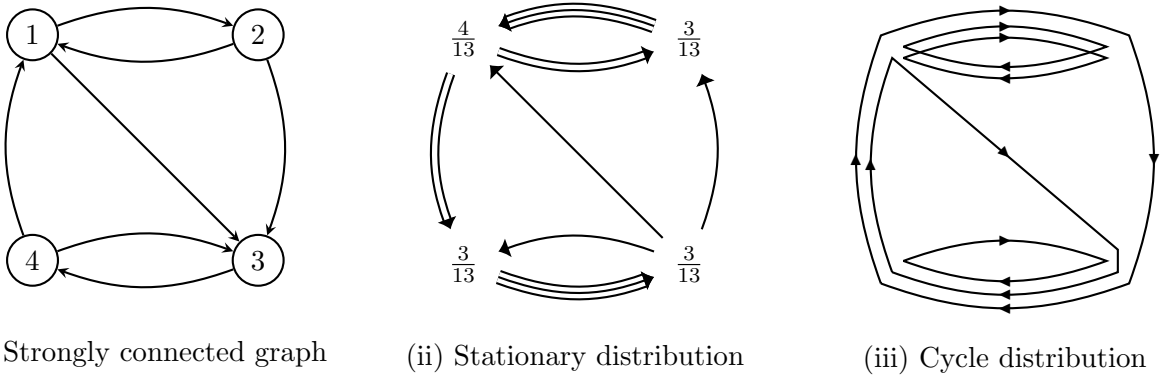

It remains to consider the case where the envy graph is not 
strongly connected. 
{In this case, we apply the key lemma to a strongly connected component $G$ with no incoming arcs, in order to avoid generating ex-ante envy for the agents outside of $G$.}
Such a component must always exist since the connectivity graph of the strongly connected components is acyclic. 
Moreover, since we only eliminate envy cycles if all the agents are envied, the selected component contains at least two agents, since otherwise, the {single} agent in the component would be unenvied.

{The distribution over the envy cycles in $G$ given by \Cref{lem:cycle_distribution} satisfies the following property: if any agent $i$ prefers some bundle $X_j$ to her own bundle $X_{i}$, then no {other} agent is more likely to get $X_j$ than $i$.
{We use this property in the form of the following lemma.}}

\begin{lemma}\label{lem:envy_cycles_sd}
    Let $(X_1', \ldots, X_n')$ be the (random) allocation obtained after eliminating an envy cycle chosen according to the probability distribution given by \Cref{lem:cycle_distribution} applied to the strongly connected component $G$. Then, for any two agents $i$ and $j$, it holds that
    \begin{align*}
        v_i(X_i') \SD v_i(X_j') \cdot \mathbb 1[X_{j}' \neq X_j].
    \end{align*}
\end{lemma}
\begin{proof}
    Let $C$ be the (random) envy cycle chosen according to the distribution given in \Cref{lem:cycle_distribution}. First, observe that for any $t \leq v_i(X_i)$, it holds that $\mathbb{P}[v_i(X_i') \geq t] = 1$ since $v_i(X_i') \geq v_i(X_i) \geq t$ by \Cref{lem:bundle_monotonicity}, and hence $\mathbb{P}[v_i(X_i') \geq t] \geq \mathbb{P}[v_i(X_j') \cdot \mathbb 1 [X_j' \neq X_j] \geq t]$.
    
    Let $t > v_i(X_i)$ and $\mathcal{K} = \{ k \in [n] : v_i(X_k) \geq t\}$.
    Note that $v_i(X_i') \geq t$ if and only if $X_i' = X_k$ for some $k \in \mathcal{K}$, i.e., it holds that $(i,k) \in C$. Similarly, 
    $v_i(X_j') \cdot \mathbb 1[X_j' \neq X_j] \geq t$ if and only if $X_j' = X_k$ for some $k \in \mathcal{K} - j$, i.e., it holds that $(j,k) \in C$.
    By the properties of $C$ given in \Cref{lem:cycle_distribution}, it holds that
    \begin{align*}
        \mathbb{P}[v_i(X_i') \geq t] &= \sum_{k \in \mathcal{K}} \mathbb{P}[(i,k) \in C]  \geq \sum_{k \in \mathcal{K}-j} \mathbb{P}[(j,k) \in C] = \mathbb{P}[v_i(X_j') \cdot \mathbb 1[X_j' \neq X_j] \geq t]
    \end{align*}
    which gives the result.
\end{proof}

\subsection{The Main Algorithm}
\label{sec:proof}

{We now present the main algorithm (\Cref{alg:fair_envy_cycles}) that combines the ideas discussed above.}

\input{algos/alg-fair_envy_cycles}

{To analyze the execution of the algorithm, we first define the execution tree of all the possible random choices that the algorithm makes on a given input.}

\begin{definition}[Execution tree]
    Let $\mathcal{H}^t$ be the collection of all the possible choices of the algorithm up to the $t$-th round, and let $\mathcal{H} = \bigcup_{t \geq 0} \mathcal{H}^t$.
    The set of vertices in the execution tree is $\mathcal{H}$ and any vertex $u \in \mathcal{H}^t$ is connected with an edge  to all the vertices $u' \in \mathcal{H}^{t+1}$ such that the first $t$ rounds in $u'$ are identical to $u$. 
\end{definition}

In particular, note that there is a single element in $\mathcal{H}^0$, and $q$ elements in $\mathcal{H}^1$, one for each allocation in the support of the random outcome of the first phase. Every node $u \in \mathcal{H}^t$ in the execution tree is naturally associated with the probability $\mathbb{P}[u]$ that the algorithm reaches the state $u$, i.e., the realization of the random choices of the algorithm in the first $t$ rounds is exactly as specified by $u$.

We say that a vertex $u \in \mathcal{H}$ is a {\em leaf} if it has no children, i.e., the algorithm terminates given the history of $u$. In particular, every leaf corresponds to a complete allocation. Let $\mathcal{L}$ be the set of all leaves of the execution tree, and let $\mathcal{L}_u$ denote the set of leaves in the subtree rooted at $u$.

For each node $u$, we denote by $(X^u_1, \dots, X^u_n)$ the allocation when the algorithm is in state $u$. For each $j$, we split the bundle $X_{j}^u$ into two parts $Y_{j}^u$ and $g^u$ in the following way which depends on whether $j$ was given an item in line~\ref{line:unenvied} or not. If the last time that $j$'s bundle changes in the history of $u$ is due to an operation in line~\ref{line:unenvied}, i.e., $j$ gets an item because $j$ is an unenvied agent, then we let $g^u$ be the item that was given to $j$ during that operation and $Y_{j}^u = X_{j}^u - g^u$. Otherwise, we let $Y_{j}^u = X_{j}^u$ and $g^u = \bot$. Note that in both cases it holds that $X_{j}^u = Y_j^u + g^u$.

{Let us first combine the observations made during the analysis of the randomization of the first and the second phase of the algorithm into the following useful lemma.}

\begin{lemma}\label{lem:sd_ineq}
    Consider a node $u \in \mathcal H$ of the execution tree, and let $u'$ be the random child of $u$ obtained after performing the next step of the algorithm. {Then,}
    \begin{align*}
        v_i(X_i^{u'}) \SD v_i(Y_{j}^{u'}) \cdot \mathbbm{1}[X_{j}^{u} \neq X_{j}^{u'}].
    \end{align*}
\end{lemma}
\begin{proof}
If $u$ is the root of the tree, the result follows from \Cref{lem:ps_prefix}. Otherwise, given the history of $u$, the algorithm either performs an envy cycle elimination step, in which case the result follows from \Cref{lem:envy_cycles_sd}, or the algorithm assigns an unallocated item to an unenvied agent, in which case either $j$ is unenvied or $j$'s bundle does not change, and the result follows.
\end{proof}

{The proof of the ex-ante fairness of the main algorithm relies on a new notion of {\em stochastic coverage} which generalizes the notion of stochastic dominance. 
In order to show the desired inequality $2 \cdot \mathbb{E}[v_i(X_i)] \geq \mathbb E[v_i(X_j)]$, we first show that $2 \cdot \mathbb{P}[v_i(X_i) \geq t] \geq \mathbb{P}[v_i(Y_j) \geq t] + \mathbb{P}[v_i(g_j) \geq t]$ which is neatly captured in terms of stochastic coverage as $(v_i(X_i), v_i(X_i)) \SC (v_i(Y_j), v_i(g_j))$. See the following definition.}

\begin{definition}[Stochastic coverage]
Given two collections $U = (x_i)_{i\in [r]}$ and $V = (y_j)_{j\in [s]}$ of non-negative random variables (not necessarily independent) for some $r,s \geq 1$, we say that $U$ stochastically covers $V$, and we write $U \SC V$, if
\begin{align*}
    \sum_{i\in [r]} \mathbb{P}[x_i \geq t] \geq \sum_{j\in [s]} \mathbb{P}[y_j \geq t] \;\;\;\;\;\; \text{for all }t > 0
\end{align*}
\end{definition}

This definition captures the standard notion of stochastic dominance between random variables by taking $r = s = 1$.
Moreover, note that {the definition requires} $t$ {to be} strictly positive, {which means, for instance, that any (even empty) collection of random variables covers any other} collection of random variables that are always equal to $0$.

\begin{remark}
An alternative way to think about stochastic coverage is in terms of a ``credit scheme''. 
{In terms of stochastic coverage,} our goal is to show that $(v_i(X_i), v_i(X_i)) \SC (v_i(Y_j), v_i(g_j))$.
The given random variables are defined on the probability space where every outcome is associated with a leaf in the execution tree. Imagine that initially each leaf $u \in \mathcal{L}$ holds two divisible tokens of mass $\mathbb{P}[u]$ each and that it can use some fraction of these tokens to compensate another leaf $v \in \mathcal{L}$ if either $v_i(X_i^u) \geq v_i(Y_j^v)$ or $v_i(X_i^u) \geq v_i(g_j^v)$. 
The statement $(v_i(X_i), v_i(X_i)) \SC (v_i(Y_i), v_i(g_j))$ is then equivalent to the existence of a credit scheme where each leaf $v \in \mathcal{L}$ receives a total mass of $\mathbb{P}[v]$ tokens to compensate $v_i(Y_j^v)$ and a total mass of $\mathbb{P}[v]$ tokens to compensate $v_i(g_j^v)$. Hence, we essentially show that the leaves where $i$ is assigned a valuable bundle can compensate all the leaves where $j$ is assigned a valuable bundle (from $i$'s perspective).
\end{remark}

The following lemma (whose proof is deferred to \Cref{sec:proofs}) establishes several useful properties of stochastic coverage.

\begin{restatable}{lemma}{lemtech}\label{lem:tech}
Stochastic coverage satisfies the following properties:
\begin{enumerate}[(i)]
    \item transitivity: for all collections $U,V,W$, if $U \SC V$ and $V \SC W$, then $U \SC W$.\label{list:tech_transitivity}
    \item concatenability: for all collections $S,T,U,V$, if $S \SC T$ and $U \SC V$, then $(S,U) \SC (T,V)$.\label{list:tech_concatenability}
    \item disjoint additivity: for all non-negative random variables $x$ and $y$ that are positive on disjoint events, it holds that $(x,y) \SC (x+y)$ and $(x+y) \SC (x,y)$.\label{list:tech_additivity}
\end{enumerate}
\end{restatable}

We also establish a way to extend stochastic coverage if the covering collection contains strictly more random variables than the covered one, which we use in the form of the following lemma. The proof is deferred to \Cref{sec:proofs}.

\begin{restatable}[Extendability]{lemma}{lemextendability}\label{lem:extendability}
    For all non-negative random variables $x,y,z$, if
    \begin{enumerate}[(i)]
        \item it holds that  $(x,x) \SC y$,
        \item there is an event $S$ such that $x$, $y$ and $z$ are non-zero only if $S$ holds, and
        \item there is some constant threshold $\delta$ such that whenever $S$ holds, it also holds that $x \geq \delta \geq z$,
    \end{enumerate}
    then it holds that $(x,x) \SC (y,z)$.
\end{restatable}

Next, we prove the following lemma which is the main ingredient in the proof of \Cref{thm:main_theorem}. The proof proceeds by induction on the execution tree, as illustrated by \Cref{fig:tree}.
 
\begin{lemma}\label{lem:sc} Letting $w \in \mathcal L$ denote the random leaf reached by the algorithm, we have 
$$
(v_i(X_i^w), v_i(X_i^w)) \SC (v_i(Y_j^w), v_i(g_j^w)).
$$
\end{lemma}

\begin{figure}[h!]
    \centering
    \input{figures/tree.tex}
    \caption{Illustration of the entire execution tree (left) and its subtrees (right). Paths on which $j$'s bundle changes are represented by plain edges, and paths on which $j$'s bundle does not change are represented by dashed edges. 
    The first step is probabilistic serial (top right) and subsequent operations include either assigning an item to an unenvied agent (middle right), or eliminating an envy cycle (bottom right).
    Let $u$ be some node, and let $u_1, \ldots, u_k$ be its children.
    Consider a leaf $w$ in the subtree rooted at $u_t$. The leaf $w$ is colored {\em red} if $X_j^w=X_j^u$, i.e., $j$'s bundle does not change after reaching $u$.
    It is colored {\em yellow} if $X_j^w=X_j^{u_t}$ and $X_j^{u_t} \neq X_j^{u}$, i.e., $j$'s bundle changes only during the first iteration after reaching $u$. 
    Finally, it is colored {\em green} if $X_j^w \neq X_j^{u_t}$, i.e., $j$'s bundle changes during one of the subsequent iterations.
    Note that a leaf's color is always with respect to some node $u$, and it can thus be colored differently depending on the chosen node $u$.}
    \label{fig:tree}
\end{figure}
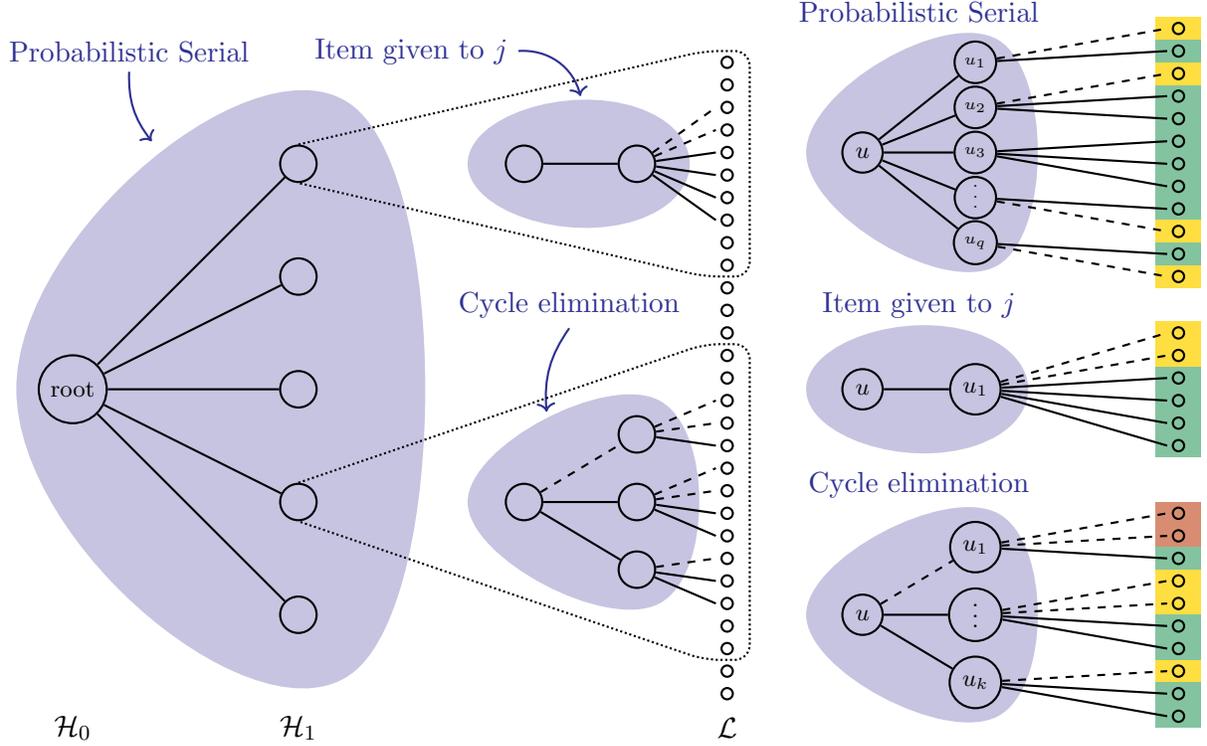

\begin{proof}
    Let $w \in \mathcal{L}$ be the random leaf reached by a random execution of the algorithm. For any node $u \in \mathcal{H}$, we define the following random variables
    \begin{align*}
        x_u &= \mathbb 1[w\in \mathcal L_u]\cdot v_i(X_i^w), \quad
        z_u &= \mathbb 1[w\in \mathcal L_u]\cdot v_i(X_i^u), \quad
        y_u &= \mathbb 1[w\in \mathcal L_u]\cdot v_i(Y_j^w) \cdot \mathbb 1[X_j^u \neq X_j^w].
    \end{align*}
    Note that $y_u$ includes $v_i(Y_j)$ from all the yellow and green leaves in \Cref{fig:tree} but not from the red ones. 
    Observe that it always holds that $x_u \geq z_u$ by \Cref{lem:bundle_monotonicity}.
    
    First, we use induction on the execution tree to show that for every node $u$ it holds that
    $$
    (x_u, x_u) \SC (y_u, z_u).\label{eq:induction}
    $$
    This statement holds for all the leaves $u \in \mathcal{L}$ since $x_u = z_u$ and $y_u = 0$. 
    
    Consider a non-leaf node $u$ and its children $u_1, \dots, u_k$. By the inductive assumption, it holds that $(x_{u_t}, x_{u_t}) \SC (y_{u_t}, z_{u_t})$ for all $1 \leq t \leq k$. This means that $v_i(Y_j)$ from the green leaves in \Cref{fig:tree} is covered by our induction hypothesis. Observe that
    \begin{align*}
      (x_u, x_u) 
      &= (x_{u_1} + \dots + x_{u_k}, x_{u_1} + \dots + x_{u_k}) && (\text{since $x_u = x_{u_1} + \ldots + x_{u_k}$}) \\
      &\SC (x_{u_1}, \ldots, x_{u_k}, x_{u_1}, \ldots, x_{u_k}) && (\text{by disjoint additivity}) \\
      &\SC (y_{u_1}, \ldots, y_{u_k}, z_{u_1}, \ldots, z_{u_k}) && (\text{by concatenability}) \\
        &\SC  (y_{u_1} + \dots + y_{u_k}, z_{u_1} + \dots + z_{u_k}) && (\text{by disjoint additivity})
    \intertext{Next, for each child $u_t$ of $u$, define}
        \mu_{u_t} &= \mathbb 1[w\in \mathcal L_{u_t}] \cdot v_i(Y_j^{u_t}) \cdot \mathbb 1[X_j^u \neq X_j^{u_t} = X_j^w] && \text{for } 1 \leq t \leq k
    \intertext{so that $y_u = y_{u_1} + \ldots + y_{u_k} + \mu_{u_1} + \ldots + \mu_{u_k}$. The variables $\mu_{u_t}$ include $v_i(Y_j)$ from the yellow leaves in \Cref{fig:tree}, which need to be covered at this induction step.
    Now, observe that}
    z_{u_1} + \dots + z_{u_k} 
    &= \sum_{t=1}^k \mathbb 1[w\in \mathcal L_{u_t}] \cdot v_i(X_i^{u_t}) && (\text{by the definition of } z_{u_t}) \\
    &\SD \sum_{t=1}^{k} \mathbb 1[w\in \mathcal L_{u_t}] \cdot v_i(Y_j^{u_t}) \cdot \mathbb 1[X_j^{u} \neq X_j^{u_t}] && (\text{by \Cref{lem:sd_ineq}}) \\
    &\geq \sum_{t=1}^{k} \mathbb 1[w\in \mathcal L_{u_t}] \cdot v_i(Y_j^{u_t}) \cdot \mathbb 1[X_j^{u} \neq X_j^{u_t} = X_j^w] \\
    &= \mu_{u_1} + \ldots + \mu_{u_k} && (\text{by the definition of } \mu_t) 
    \intertext{Combining the statements above gives}
        (x_u,x_u) 
        &\SC (y_{u_1} + \ldots + y_{u_k}, z_{u_1} + \ldots + z_{u_k}) \\
        &\SC (y_{u_1} + \ldots + y_{u_k}, \mu_{u_1} + \ldots + \mu_{u_k}) \\
        &\SC y_{u_1} + \ldots + y_{u_k} + \mu_{u_1} + \ldots + \mu_{u_k} && (\text{by disjoint additivity}) \\
        &= y_u
    \end{align*}
    so that $v_i(Y_j)$ from all the green and yellow leaves is covered.
    To complete the inductive proof, we use the extendability of stochastic coverage (\Cref{lem:extendability}) with the event $S = \{w \in \mathcal{L}_u \}$ and the threshold
    $\delta = v_i(X_i^u)$ to get
    \begin{align*}
        (x_u, x_u) \SC (y_u, z_u)
    \end{align*}
    which completes the induction step.
    
    Now to prove the statement of the lemma, it remains to account for $v_i(g_j^w)$. By disjoint additivity, we get
    \begin{align*}
        (x_\text{root}, x_\text{root}) \SC (y_\text{root}, z_\text{root}) \SC y_\text{root}
    \end{align*}
    since $z_\text{root} = 0$. Note that by strong separation (\Cref{lem:ps_strong_sep}), it holds that $x_\text{root} \geq \max_{q \in \mathcal{L}_\text{root}} v_i(g_j^q) \geq v_i(g_j^w)$. Hence, we use extendability (\Cref{lem:extendability}) once again, now with the full event $S = \{w \in \mathcal{L}_\text{root}\}$ and the threshold $\delta = \max_{q \in \mathcal{L}_\text{root}} v_i(g_j^q)$, to get
    \begin{align*}
    (x_\text{root}, x_\text{root}) 
            &\SC (y_\text{root}, v_i(g_j^w)) 
    \end{align*}
    which proves the lemma.
\end{proof}

We are now ready to prove the main theorem.

\begin{proof}[Proof of \Cref{thm:main_theorem}]
The allocation of \Cref{alg:fair_envy_cycles} is ex-post $\frac{1}{2}$-EFX and ex-post EF1 by \Cref{lem:ex_post}.
To show that the outcome is ex-ante $\frac{1}{2}$-EF, observe that
\begin{align*}
    2 \cdot \mathbb{E}[v_i(X_i)] &= \int_{t \geq 0} \mathbb{P} [v_i(X_i) \geq t ] + \mathbb{P} [v_i(X_i) \geq t ] \;\text{d}t && (\text{by the properties of expectation})\\
         &\geq  \int_{t \geq 0} \mathbb{P} [v_i(Y_i) \geq t ] + \mathbb{P}[v_i(g_j) \geq t] \;\text{d}t && (\text{by \Cref{lem:sc})}\\
         &= \mathbb{E}[v_i(Y_j)] + \mathbb{E}[v_i(g_j)] && (\text{by the properties of expectation}) \\
         &= \mathbb{E}[v_i(Y_j) + v_i(g_j)] && (\text{by additivity of expectation}) \\
         &\geq \mathbb{E}[v_i(X_j)] && (\text{by subadditivity of }v_i)
\end{align*}  
which is the desired inequality.

Finally, we prove that \Cref{alg:fair_envy_cycles} runs in polynomial time. First, recall that probabilistic serial can be implemented in polynomial time \cite{BM01}, and that the resulting fractional allocation can be rounded in polynomial time (\Cref{thm:bvn}). Second, each iteration of our randomized algorithm can be carried out in polynomial time by \Cref{lem:cycle_distribution}. Finally, the deterministic envy cycles procedure terminates after polynomially many steps (\Cref{lem:cycles-polytime}), and since we only randomize choices that are made arbitrarily in this deterministic procedure, the polynomial running time applies to every outcome of our randomized algorithm.
\end{proof}

\section{The Case of Two Agents}
\label{sec:two-agents}

In this section, we provide tradeoffs between ex-post EFX and ex-ante EF for the case of two agents. 
Our results are demonstrated in \Cref{fig:tradeoffs} (left).

In our algorithms and analysis we use the notion of a {\em $\beta$-EFX partition}, defined as follows.

\begin{definition}[$\beta$-EFX partition]
A partition of the items $X = (X_1, X_2)$ is a $\beta$-EFX partition for agent $i$ if $v_i(X_1) \geq \beta \cdot v_i(X_2 - g)$ for all $g \in X_2$ and $v_i(X_2) \geq \beta \cdot v_i(X_1 - g)$ for all $g \in X_1$.
\end{definition}

Clearly, given a $\beta$-EFX partition $(X_1,X_2)$ for agent $i$, we can allocate either $X_1$ or $X_2$ to agent $i$ and preserve $\beta$-EFX with respect to agent $i$. 
{We now provide the following construction of a $\frac{1}{2}$-EFX partition for subadditive agents that  we use to obtain improved tradeoffs for the case of two agents (\Cref{prop:subadditive-2-agents-ef}). In \Cref{sec:subadd_appendix} we show that this lemma is tight.}

\begin{lemma}[Partition for subadditive]\label{lem:2-agent-subadd}
Assume that the valuations are subadditive. For any agent $i \in \{1,2\}$, let $X = (X_1, X_2)$ be a partition of items with $v_i(X_1) \geq v_i(X_2)$ that minimizes $|v_i(X_1) - v_i(X_2)|$, and, subject to this, minimizes $|X_1|$. Then, $X$ is a $\frac{1}{2}$-EFX partition for agent~$i$.
Moreover, there exists an instance with submodular (and hence subadditive) valuations such that this partition is not a $\beta$-EFX partition for agent~$i$ for any $\beta > 1/2$.
\end{lemma}

\begin{proof}
    For the purpose of contradiction, assume that there is some $g \in X_1$ so that $v_i(X_1 - g) > 2 \cdot v_i(X_2)$.
    First, suppose that $v_i(X_2+g) \leq v_i(X_1 - g)$. Then, 
    $$v_i(X_2) \leq v_i(X_2+g) \leq v_i(X_1-g) \leq v_i(X_1).$$
    This implies that
    $$v_i(X_1-g) - v_i(X_2+g) \leq v_i(X_1) - v_i(X_2),$$ 
    contradicting the assumption 
    since $|X_1 - g| < |X_1|$.
    
    Next, suppose that $v_i(X_2+g) > v_i(X_1 - g)$. Observe that
    \begin{align*}
        v_i(X_1-g) - v_i(X_2) &> v_i(X_2) && (\text{since }2 \cdot v_i(X_2) < v_i(X_1 - g)) \\
        &\geq v_i(\items) - v_i(X_1) && (\text{since } v_i(X_1) + v_i(X_2) \geq v_i(\items)) \\
        &\geq v_i(X_2+g) - v_i(X_1) && (\text{since }v_i(X_2+g) \leq v_i(\items))
    \end{align*}
    which implies that
    \begin{align*}
        &v_i(X_2 + g) - v_i(X_1) < v_i(X_1-g) - v_i(X_2)  \\
        \Longleftrightarrow\quad &v_i(X_2+g)-v_i(X_1-g) < v_i(X_1) - v_i(X_2)\\
        \Longleftrightarrow\quad &|v_i(X_2+g)-v_i(X_1-g)| < |v_i(X_1) - v_i(X_2)| && (\text{since } v_i(X_2 + g) \geq v_i(X_1 - g)),
    \end{align*}
    contradicting the assumption that $X$ is chosen to minimize $|v_i(X_1) - v_i(X_2)|$.
\end{proof}

The following propositions provide fairness tradeoffs for two subadditive valuations.

We first present an algorithm that finds an allocation that is ex-ante EF and ex-post $\frac{1}{2}$-EFX.

\begin{proposition}
Every instance with two subadditive valuations admits a randomized allocation that is ex-ante EF and ex-post $\frac{1}{2}$-EFX.
\label{prop:subadditive-2-agents-ef}
\end{proposition}

\begin{proof}
    We may assume without loss of generality, by uniformly scaling the valuation functions of the agents, that each agent's value for the bundle containing all items is exactly 1. Suppose agent~1 and agent~2 choose $\frac{1}{2}$-EFX partitions $A=(A_1,A_2)$ and $B=(B_1,B_2)$ respectively, as specified in \Cref{lem:2-agent-subadd}.
    If agent~$2$ prefers $A_2$ over $A_1$ or agent~$1$ prefers $B_2$ over $B_1$ then there is a deterministic EF allocation.
    
    For the remaining case, suppose both agents prefer $A_1$ to $A_2$ and $B_1$ to $B_2$. Consider the random allocation which, with probability $\frac{1}{2}$, gives $A_2$ to agent~1 and $A_1$ to agent~2, and with probability $\frac{1}{2}$ gives $B_1$ to agent~1 and $B_2$ to agent~2. By construction, both allocations in the support are EFX for both agents.

    Agent~1's expected value in this random allocation is $(1/2)\cdot(v_1(B_1)+v_1(A_2))$. The expected value that agent~1 has for agent~2's bundle is $(1/2) \cdot (v_1(A_1)+v_1(B_2))$. Since $A$ is chosen to minimize $|v_1(A_1)-v_1(A_2)|$, we have
    \begin{align*}
        &v_1(B_1)-v_1(B_2) \geq v_1(A_1)-v_1(A_2) \\
\quad\Longleftrightarrow\quad &(1/2)\cdot\left(v_1(B_1)+v_1(A_2)\right) \geq (1/2)\cdot \left(v_1(A_1)+v_1(B_2)\right)
    \end{align*}
    so the allocation is ex-ante EF for agent~1. A similar argument holds for agent~2.
\end{proof}

We next give an algorithm that finds an allocation that is ex-ante $\frac{2}{3}$-EF and ex-post EFX.

\begin{proposition}
Every instance with two subadditive valuations admits a randomized allocation that is ex-ante $\frac{2}{3}$-EF and ex-post EFX.
\label{prop:subadditive-2-agents-efx}
\end{proposition}

\begin{proof}
    Once again, we may assume without loss of generality, by uniformly scaling the valuation functions of the agents, that each agent's value for the bundle containing all items is exactly 1. Suppose agent~1 chooses any EFX partition $A = (A_1, A_2)$ and agent~2 chooses any EFX partition $B = (B_1, B_2)$. We may assume without loss of generality that $v_1(A_1) \geq v_1(A_2)$ and $v_2(B_1) \geq v_2(B_2)$. Again, if agent~$2$ prefers $A_2$ over $A_1$ or agent~$1$ prefers $B_2$ over $B_1$, there is a deterministic EF allocation.
    
    For the remaining case, suppose both agents prefer $A_1$ to $A_2$ and $B_1$ to $B_2$. Consider the random allocation which, with probability ${1}/{2}$, gives $A_2$ to agent~1 and $A_1$ to agent~2, and with probability ${1}/{2}$ gives $B_1$ to agent~1 and $B_2$ to agent~2. Clearly, both allocations in the support are $EFX$ for both agents.

    Agent~1's expected value in this random allocation is $({1}/{2})\cdot(v_1(B_1)+v_1(A_2))$. The expected value that agent~1 has for agent~2's bundle is $({1}/{2}) \cdot (v_1(A_1)+v_1(B_2))$. By assumption, we have $v_1(B_1) \geq v_1(B_2)$, and by subadditivity, $v_1(B_1) + v_1(B_2) \geq v_1(B_1\cup B_2) = 1$. Consequently, we have $v_1(B_1) \geq {1}/{2}$. Similarly, we have $v_1(A_1) \geq {1}/{2}$.

    If $v_1(A_1) > 2 \cdot v_1(A_2)$, then by subadditivity $v_1(A_1) > 1/{2}$. Additionally, the set $A_1$ contains only a single item: suppose for a contradiction that $A_1$ contains at least two items, and let $g_1$ and $g_2$ be two items in $A_1$. By subadditivity, $v_1(A_1 - g_1) + v_1(A_1 - g_2) \geq v_1(A_1)$, thus at least one of these items (say $g_1$) is such that $v_1(A_1 - g_1) \geq ({1}/{2}) \cdot v_1(A_1) > v_1(A_2)$, violating the assumption that $A$ is an EFX partition for agent~1. Consequently, $A_1$ contains a single item of value greater than ${1}/{2}$. By monotonicity, this item is in $B_1$ and not $B_2$, hence $A_1 \subseteq B_1$ and $B_2 \subseteq A_2$. Thus we have $v_1(B_1) \geq v_1(A_1)$ and $v_1(A_2) \geq v_1(B_2)$, so this random allocation is ex-ante EF for agent~1.

    It remains to consider agent~1's envy in the case where $v_1(A_1) \leq 2 \cdot v_1(A_2)$. We have
    \begin{align*}
        \frac{v_1(A_2) + v_1(B_1)}{v_1(A_1) + v_1(B_2)} &\geq \frac{v_1(A_2) + v_1(B_1)}{v_1(A_1) + v_1(B_1)}
        \geq \frac{(1/2) \cdot v_1(A_1) + v_1(B_1)}{v_1(A_1) + v_1(B_1)}
        \geq \frac{(1/2) \cdot v_1(A_1) + (1/2)}{v_1(A_1) + (1/2)}
    \end{align*}
    Since $v_1(A_1) \leq 1$ the above ratio is at least $\frac{2}{3}$, so agent~1 is at least ex-ante $\frac{2}{3}$-EF. A similar argument applies to the bundle values for agent 2.
\end{proof}

Finally, the following proposition establishes some impossibility results, showing, among other results, that \Cref{prop:subadditive-2-agents-efx} is tight.

\begin{proposition}
\label{prop:ub-subadditive}
For any $0.618 \approx \varphi-1 < \beta \leq 1$, there exists an instance with no randomized allocation that is simultaneously ex-ante $\alpha$-EF and ex-post $\beta$-EFX, where $\alpha = \frac{\beta+1}{\beta^2+2\beta}$.
\end{proposition}

\begin{proof}
Consider the following instance with three items, $a$, $b$ and $c$. Take $0 < \varepsilon < 1-2\beta$.
\begin{center}
    \begin{tabular}{||c | c c c c c c c c||} 
     \hline
     & $\emptyset$ & $a$ & $b$ & $c$ & $ab$ & $ac$ & $bc$ & $abc$ \\ [0.5ex] 
     \hline\hline
     $v_1$ & $0$ & $1+\varepsilon$ & $\beta$ & $\beta$ & $1+\varepsilon$ & $1+\varepsilon$ & $2\beta$ & $2\beta$ \\ 
     \hline
     $v_2$ & $0$ & $\beta$ & $1+\varepsilon$ & $1+ \varepsilon$ & $1+\beta$ & $1+\beta$ & $1+\varepsilon$ & $1+\beta$ \\
     \hline
    \end{tabular}
\end{center}
It can be verified that $v_1$ and $v_2$ are subadditive (but observe that $v_1$ is not submodular because $v_1(abc)-v_1(ab) > v_1(ac)-v_1(a)$).
There are three deterministic $\beta$-EFX allocations, namely $X^1 = (a,bc)$, $X^2 = (ab,c)$ and $X^3 = (ac,b)$. Assume that the corresponding probabilities are $p_1$, $p_2$ and $p_3$. We have
\begin{align*}
\mathbb E[v_1(X_1)] &= p_1\cdot (1+\varepsilon) + p_2 \cdot (1+\varepsilon) + p_3 \cdot (1+\varepsilon) = 1+\varepsilon\\
\mathbb E[v_2(X_2)] &= p_1\cdot (1+\varepsilon) + p_2 \cdot (1+\varepsilon) + p_3 \cdot (1+\varepsilon) = 1+\varepsilon\\
\mathbb E[v_1(X_2)] &= p_1\cdot 2\beta + p_2 \cdot \beta + p_3 \cdot \beta = \beta\cdot(1+p_1)\\
\mathbb E[v_2(X_1)] &= p_1\cdot \beta + p_2 \cdot (1+\beta) + p_3 \cdot (1+\beta) =1+\beta-p_1
\end{align*}
In particular, observe that the maximum envy is minimized when $\beta\cdot (1+p_1) = 1+\beta - p_1$, that is, when $p_1 = 1/(\beta+1)$, in which case we have $\mathbb E[v_1(X_2)] = \mathbb E[v_2(X_1)] = 1/\alpha$, with $\alpha = \frac{\beta+1}{\beta^2+2\beta}$.
\end{proof}

As a direct corollary, we get that: (i) 
there is no random allocation that is ex-ante $\alpha$-EF and ex-post EFX for $\alpha>\frac{2}{3}$, and (ii) there is no random allocation that is ex-ante EF and ex-post $\beta$-EFX for $\beta > \varphi-1$.

In \Cref{app:two-agents} we show that \Cref{prop:subadditive-2-agents-efx} is tight even with respect to two submodular agents (see \Cref{prop:2agent-submod-tight}).
Moreover, for general (monotone) valuations, we show that no approximate guarantees are possible (see \Cref{prop:ub-general}).

\section{Concluding Remarks}

Our results significantly advance the state-of-the-art in best-of-both-worlds fairness. We demonstrate that strong ex-ante and ex-post fairness guarantees can be achieved even in settings with subadditive valuations and even with respect to the stronger envy-freeness notion of EFX. Our work suggests several natural problems for future research.

First, while our analysis is tight for our algorithm (see \Cref{ex:tight_analysis}), the existence of an allocation that is ex-ante $\alpha$-EF and ex-post $\frac{1}{2}$-EFX for any $\alpha > \frac{1}{2}$ remains open.

Second, many of the existing results on EFX rely on the envy cycles procedure, e.g.,
the existence of a ($\varphi-1$)-EFX allocation for any number of additive agents \cite{AMN20}.
An interesting direction to investigate is whether our techniques can be combined with these results to obtain best-of-both-worlds guarantees for these settings.

Third, while our algorithm produces an execution tree of polynomial depth, and it computes the random allocation in polynomial time, the resulting distribution might have exponential-size support. By Carath\'eodory's theorem, it is possible to reduce the support size of the final distribution to $n\cdot(n-1)+1$. However, it remains a challenge to construct this polynomial-size distribution in polynomial time. Unfortunately, techniques similar to those of~\citet{FSV20} (who reduce the support size after each iteration) and~\citet{A20} do not apply easily in our setting.

Finally, it would be intriguing to explore whether our positive results can be augmented with any type of efficiency guarantees. Prior work shows that random allocations that are ex-ante fractionally-PO, ex-ante EF, and ex-post EF1 don't exist even for additive valuations~\cite{FSV20}. Similarly, there are instances with additive valuations that admit no allocations that are EFX and PO~\cite{PR18}. It is plausible that stronger impossibility results (for weaker notions of efficiency) apply in our setting.

\bibliographystyle{plainnat}
\bibliography{bib}

\appendix

\section{Omitted Theorems and Proofs for Two Agents}
\label{app:two-agents}

\subsection{Additive Valuations}

\begin{lemma}[Partition for additive]\label{lem:2-agent-add}
Assume that the valuations are additive. For any agent $i \in \{1,2\}$, let $X = (X_1, X_2)$ be a partition of the items with $v_i(X_1) \geq v_i(X_2)$ that minimizes $|v_i(X_1) - v_i(X_2)|$ and, subject to this, minimizes $|X_1|$. Then, $X$ is an EFX partition for agent~$i$.
\end{lemma}

\begin{proof}
Suppose that $v_i(X_1) \geq v_i(X_2)$ and that it holds that $v_i(X_1 - g) > v_i(X_2)$ for some $g \in X_1$, i.e., $X$ is not an EFX partition for agent $i$. Consider the partition $X' = (X_1 - g, X_2 + g)$. 
If $v_i(g) = 0$, then this contradicts the assumption since $|X_1 - g| < |X_1|$.
Note that 
$$v_i(X_1') - v_i(X_2') = v_i(X_1 - g) - v_i(X_2 + g) = v_i(X_1) - v_i(X_2) - 2 \cdot v_i(g)$$ 
and
$$0 < v_i(g) < v_i(X_1) - v_i(X_2).$$
It follows that
$$ -v_i(X_1) + v_i(X_2) < v_i(X_1') - v_i(X_2') < v_i(X_1) - v_i(X_2),$$
which contradicts the assumption that $X$ is chosen to minimize $|v_i(X_1) - v_i(X_2)|$.
\end{proof}

The following proposition provides fairness guarantees for two additive valuations.

\begin{proposition}
    Every instance with two additive valuations admits a randomized allocation that is ex-ante EF and ex-post EFX.
\label{prop:additive-2-agents}
\end{proposition}

\begin{proof}
    Let agents $1$ and $2$ choose EFX partitions $A = (A_1, A_2)$ and $B = (B_1,B_2)$ respectively, as specified in \Cref{lem:2-agent-add}.
    Observe that if agent $2$ prefers $A_2$ over $A_1$ or agent $1$ prefers $B_2$ over $B_1$, then there is a deterministic EF allocation. 
    
    Assume that both agents prefer $A_1$ over $A_2$ and $B_1$ over $B_2$.
    Consider the following randomized allocation. With probability $1/2$ give $A_1$ to $2$ and $A_2$ to $1$, and with probability $1/2$ give $B_1$ to $1$ and $B_2$ to $2$. Both of the deterministic allocations are EFX by the assumption.

    The expected value of agent $1$ for the allocation that agent $1$ gets is $(1/2) \cdot (v_1(B_1) + v_1(A_2))$. The expected value of agent $1$ for the allocation that agent $2$ gets is $(1/2) \cdot (v_1(A_1)+v_1(B_2))$. Since $A$ is chosen to minimize $|v_1(A_1) - v_1(A_2)|$, it holds that
    $$ v_1(A_1)-v_1(A_2) \leq v_1(B_1)-v_1(B_2) $$
    and so
    $$ (1/2) \cdot (v_1(A_1)+v_1(B_2)) \leq (1/2) \cdot (v_1(B_1) + v_1(A_2))  $$
    which implies that agent $1$ is ex-ante EF. The same argument applies to agent $2$.
\end{proof}

\subsection{Subadditive Valuations}\label{sec:subadd_appendix}

\begin{proof}[Proof of tightness of \Cref{lem:2-agent-subadd}]
   To see that the partition specified in \Cref{lem:2-agent-subadd} cannot provide better guarantees, consider the following submodular valuation over items $a,b,c$, for some $\varepsilon > 0$.
    \begin{center}
        \begin{tabular}{||c | c c c c c c c c||} 
         \hline
         & $\emptyset$ & $a$ & $b$ & $c$ & $ab$ & $ac$ & $bc$ & $abc$ \\ [0.5ex] 
         \hline\hline
         $v_1$ & $0$ & $2-\varepsilon$ & $2-\varepsilon$ & $1$ & $2-\varepsilon$ & $3-\varepsilon$ & $3-\varepsilon$ & $3-\varepsilon$ \\ 
         \hline
        \end{tabular}
    \end{center}
    
   One can easily verify that  the unique partition $X = (X_1,X_2)$ that minimizes $|v_1(X_1) - v_1(X_2)|$ is given by $X_1 = \{a, b\}$ and $X_2 = \{c\}$. 
   
   It holds that 
    $$v_1(X_1 - b) = v_1(a) = 2-\varepsilon = (2-\varepsilon) \cdot v_1(c) = v_1(X_2)$$ 
    which means that $X$ is not an $\beta$-EFX partition for any $\beta > \frac{1}{2-\varepsilon}$. The result follows by talking $\varepsilon \to 0$.
\end{proof}

\subsection{Submodular Valuations}

The following proposition shows that \Cref{prop:subadditive-2-agents-efx} is tight even for submodular valuations.

\begin{proposition}
\label{prop:2agent-submod-tight}
    For any $\alpha > 2/3$, there is an instance with two submodular valuations that admits no randomized allocation that is ex-ante $\alpha$-EF and ex-post EFX.
\end{proposition}

\begin{proof}
    Consider the following instance with three items, $a,b,c$, for some $\varepsilon > 0$.
    \begin{center}
        \begin{tabular}{||c | c c c c c c c c||} 
         \hline
         & $\emptyset$ & $a$ & $b$ & $c$ & $ab$ & $ac$ & $bc$ & $abc$ \\ [0.5ex] 
         \hline\hline
         $v_1$ & $0$ & $1/2 + \varepsilon$ & $1/2$ & $1/2$ & $1$ & $1/2+\varepsilon$ & $1$ & $1$ \\ 
         \hline
         $v_2$ & $0$ & $1/2$ & $1/2+ \varepsilon$ & $1/2$ & $1$ & $1$ & $1/2+\varepsilon$ & $1$ \\
         \hline
        \end{tabular}
    \end{center}
    It can be verified that $v_1$ and $v_2$ are submodular.

    Observe that there are only two deterministic EFX allocations, namely $X^1 = (a,bc)$ and $X^2 = (ac,b)$. Consider a lottery $X$ which returns $X^1$ with probability $p$ and $X^2$ with probability $1-p$. We compute 
    \begin{align*}
    \mathbb{E}[v_1(X_1)] &= p \cdot v_1(X_1^1) + (1-p) \cdot v_1(X_1^2) \\
    &= p \cdot v_1(a) + (1-p) \cdot v_1(ac) \\
    &= p \cdot (1/2 + \varepsilon) + (1-p) \cdot (1/2 + \varepsilon) \\
    &= 1/2 + \varepsilon \\
    \mathbb{E}[v_1(X_2)] &= p \cdot v_1(X_2^1) + (1-p) \cdot v_1(X_2^2) \\
    &= p \cdot v_1(bc) + (1-p) \cdot v_1(b) \\
    &= p \cdot 1 + (1-p) \cdot (1/2)
    \intertext{and so by symmetry,}
    \mathbb{E}[v_2(X_1)] &= p \cdot (1/2) + (1-p) \cdot 1 \\
    \mathbb{E}[v_2(X_2)] &= 1/2 + \varepsilon.
    \end{align*}
    Note that
    \begin{align*}
    \max\big(\mathbb{E}[v_1(X_2)], \mathbb{E}[v_2(X_1)]\big) = 
        \max\big(p \cdot 1 + (1-p) \cdot (1/2), p \cdot (1/2) + (1-p) \cdot (1/2)\big) \geq 3/4
    \end{align*}
    and so at least one of the agents is not $\alpha$-EF for any $ \alpha > \frac{1/2 + \varepsilon}{3/4}$. The result follows by taking $\varepsilon \to 0$.
\end{proof}

\subsection{General Monotone Valuations}
\label{sec:general-monotone-imposs}

{For general monotone valuations, obtaining any approximate fairness guarantees is hopeless.
This is cast in the following proposition.}

\begin{proposition}\label{prop:ub-general}
    There exists an instance with two monotone valuations that admits no randomized allocation that is ex-ante $\alpha$-EF and ex-post $\beta$-EFX, for any $\alpha > 0$ and $\beta > 0$.
\label{prop:impossibility-general}
\end{proposition}

\begin{proof}
    Consider the following instance with three items, $a$, $b$, and $c$. Let $K > 0$.
    \begin{center}
        \begin{tabular}{||c | c c c c c c c c||} 
         \hline
         & $\emptyset$ & $a$ & $b$ & $c$ & $ab$ & $ac$ & $bc$ & $abc$ \\ [0.5ex] 
         \hline\hline
         $v_1$ & $0$ & $1$ & $0$ & $0$ & $1$ & $1$ & $K$ & $K$ \\ 
         \hline
         $v_2$ & $0$ & $0$ & $1$ & $0$ & $1$ & $K$ & $1$ & $K$ \\
         \hline
        \end{tabular}
    \end{center}
    It can be verified that $v_1$ and $v_2$ are monotone.

    Observe that there are only two deterministic $\beta$-EFX allocations (for any $\beta > 0$), namely $X^1 = (a,bc)$ and $X^2 = (ac,b)$. Consider a lottery $X$ which returns $X^1$ with probability $p$ and $X^2$ with probability $1-p$. We compute 
    \begin{align*}
    \mathbb{E}[v_1(X_1)] &= p \cdot v_1(X_1^1) + (1-p) \cdot v_1(X_1^2) \\
    &= p \cdot v_1(a) + (1-p) \cdot v_1(ac) \\
    &= p \cdot 1 + (1-p) \cdot 1 \\
    &= 1 \\
    \mathbb{E}[v_1(X_2)] &= p \cdot v_1(X_2^1) + (1-p) \cdot v_1(X_2^2) \\
    &= p \cdot v_1(bc) + (1-p) \cdot v_1(b) \\
    &= p \cdot K + (1-p) \cdot 0
    \intertext{and so by symmetry,}
    \mathbb{E}[v_2(X_1)] &= p \cdot 0 + (1-p) \cdot K \\
    \mathbb{E}[v_2(X_2)] &= 1.
    \end{align*}
    Note that
    \begin{align*}
    \max\big(\mathbb{E}[v_1(X_2)], \mathbb{E}[v_2(X_1)]\big) = 
        \max\big(p \cdot K, (1-p) \cdot K\big) \geq K/2
    \end{align*}
    and so at least one of the agents is not $\alpha$-EF for any $ \alpha > 2/K$. The result follows by taking $K \to \infty$.
\end{proof}

\section{The Deterministic Envy Cycles Procedure}
\label{sec:det_envy_cycles}

{In this section, we expand the description of the deterministic envy cycles procedure (\Cref{alg:det_envy_cycles}).}

\input{algos/det_envy_cycles}

\lembundlemonotonicity*
\begin{proof}
    In each step of the second phase, agent $i$'s bundle can only be modified in one of the following two ways: either agent~$i$ is assigned a new item, in which case agent~$i$'s value for her bundle weakly increases by monotonicity, or agent~$i$ is involved in a cycle elimination step, in which case agent~$i$ gets a bundle that she envied before the cycle elimination step.
\end{proof}

\lemexpost*
\begin{proof}
    Since any matching gives one item to each agent, the partial allocation at the end of the first phase is EFX (which is a stronger property than both $\frac{1}{2}$-EFX and EF1).

    We first show that assigning an unallocated item $g$ to an unenvied agent $j$ does not violate $\frac{1}{2}$-EFX or EF1. Indeed, since $j$ is unenvied, it holds that $v_i(X_i) \geq v_i(X_j)$ for any other agent $i$. 
    Additionally, by weak separation, we have $v_i(X_i) \geq v_i(g)$. 
    Hence, for any item $h \in X_i + g - h$, it holds that
    \begin{align*}
        v_i(X_i) &\geq (1/2) \cdot (v_i(X_i) + v_i(g)) && (\text{by the above}) \\
        &\geq (1/2) \cdot v_i(X_i+g) && (\text{by subadditivity}) \\
        &\geq (1/2) \cdot v_i(X_i+g-h) && (\text{by monotonicity})
    \end{align*}
    which means that $\frac{1}{2}$-EFX still holds, and 
    \begin{align*}
        v_i(X_i) &\geq v_i(X_j) = v_i((X_j + g) - g)
    \end{align*}
    which means that EF1 still holds. 

    We now show that eliminating an envy cycle preserves EF1 and $\frac{1}{2}$-EFX. Indeed, after this step, the value of any agent for her own bundle can only increase, and the set of all assigned bundles remains the same. The result follows.
\end{proof}

\lemcyclespolytime*

\begin{proof}
    To prove this, we need to bound the number of steps in the second phase. Recall that in each step, the algorithm either (1) assigns an unallocated item to an unenvied agent, or (2) eliminates an envy cycle.
    Note that the number of steps where operation (1) is executed is at most $m$ since the number of unallocated items strictly decreases with every operation (1).

    We now show that after every $n^2$ consecutive steps where operation (2) is executed, there must be a step where operation (1) is executed. It then follows that the number of steps where operation (2) is executed is at most $n^2m$. To prove the claim, consider a step that eliminates a cycle $C$ in the envy graph. 
    As a result of this step, the arcs in the cycle $C$ disappear from the envy graph. Moreover, every other arc not in $C$ either disappears or is shifted over by one agent in the envy graph. Thus the number of arcs in the envy graph strictly decreases with each cycle elimination step. This proves the claim since the number of arcs in the envy graph is less than $n^2$.
\end{proof}

\section{The Probabilistic Serial Lottery}
\label{sec:ps_explained}
     
{In this section we describe the Probabilistic Serial procedure by analyzing the PS-Lottery algorithm of \citet{A20}. We show that while the algorithms of \citet{A20} and \citet{FSV20} produce outcomes satisfying ex-post EF1, neither of them provides any ex-post EFX guarantees.
The following running example is used to illustrate these points.}

\begin{example}\label{ex:EF1notEFX}
Consider the following example, with 2 agents having additive valuations {over four items, as described in the table below.}
    \begin{center}
        \begin{tabular}{||c | c c c c||} 
         \hline
         & apple & banana & celery & durian \\ [0.5ex] 
         \hline\hline
         $v_1$ & $10$ & $3\varepsilon$ & $2\varepsilon$ & $0$ \\
         $v_2$ & $10$ & $2\varepsilon$ & $3\varepsilon$ & $0$ \\
         \hline
        \end{tabular}
    \end{center}
\end{example}

{
The PS-Lottery algorithm works as follows. 
First, make the number of agents and items equal, by possibly adding some dummy items (of zero value), and then creating {$k={m}/{n}$} copies of each agent. In \Cref{ex:EF1notEFX}, the required number of copies is $k=2$ copies per agent. Denote by $1_1$ and $1_2$ the two copies of agent 1, and by $2_1$ and $2_2$ the two copies of agent 2. 
\blfootnote{images: \url{flaticon.com}}

Second, simulate the eating procedure of~\citet{BM01} for {$k$ units of time}.
{Initially, all items are unconsumed, and during the eating process, items get fractionally consumed by agents.
For every $1 \leq t \leq k$, during the $t$-th unit of time, the $t$-th copy of each agent}
participates {in the eating process} and consumes her favourite available item at a constant rate of one item per one unit of time.
Notably, multiple agents may be consuming the same item simultaneously.
If at any point, an item is fully consumed, then each one of the agents who are currently eating this item switches to her respective next-favorite item. 
Here, if an agent is indifferent between {multiple} items, she selects an arbitrary one.
Note that after 
$k$ {units of time}, every item is fully consumed. 
}

\Cref{fig:ps_eating} illustrates 
the eating procedure described above {by showing (a) the intermediate state} {after one unit of time, and (b) {the final outcome} after two units of time}. 
The outcome can be thought of as a {\em fractional} allocation of the items to the copied agents, i.e., a real matrix $(Z_{i_r,j})_{i\in [n], r \in [k], j \in [m]}$ where the $r$-th copy of agent $i$ receives a $Z_{i_r, j}$ fraction of item $j$.

The PS-Lottery algorithm then decomposes the fractional allocation into a lottery over {\em integral} allocations using the Birkhoff-von Neumann theorem (see \Cref{sec:matching_distribution}). 
The resulting lottery for the running example is given below:

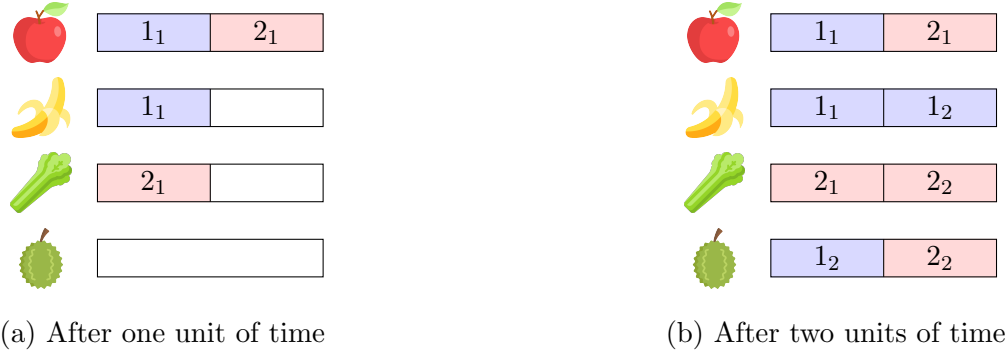
\begin{figure}[t!]
\input{figures/eating}
\caption{The eating procedure of Probabilistic Serial, applied to the instance of \Cref{ex:EF1notEFX}.}
\label{fig:ps_eating}
\end{figure}

\[
\begin{pNiceMatrix}[first-row, first-col]
      & a & b & c & d  \\
  1_1 & 0.5 & 0.5 & 0 & 0  \\
  1_2 & 0 & 0.5 & 0 & 0.5  \\
  2_1 & 0.5 & 0 & 0.5 & 0  \\
  2_2 & 0 & 0 & 0.5 & 0.5 \\
\end{pNiceMatrix}
 = 0.5 \cdot 
 \begin{pNiceMatrix}[first-row, first-col]
    & a & b & c & d  \\
  & 1 & 0 & 0 & 0  \\
  & 0 & 1 & 0 & 0  \\
  & 0 & 0 & 1 & 0  \\
  & 0 & 0 & 0 & 1 \\
\end{pNiceMatrix}
+ 0.5 \cdot 
\begin{pNiceMatrix}[first-row, first-col]
& a & b & c & d  \\
  & 0 & 1 & 0 & 0  \\
  & 0 & 0 & 0 & 1  \\
  & 1 & 0 & 0 & 0  \\
  & 0 & 0 & 1 & 0 \\
\end{pNiceMatrix}.
 \]
 
 {Here, the fractional allocation obtained via the eating procedure is given on the left-hand side, and it can be implemented by a lottery over the two integral allocations given in the middle and right-hand side matrices, each with probability $1/2$.} 
 
{The algorithm then samples one integral allocation according to the obtained probability distribution over integral allocations, and combines the items assigned to all copies of agent $i$ into one bundle that is assigned to agent $i$ in the final allocation. Hence, in the running example, the final allocation is {either the allocation}
where agent 1 receives the bundle $\{a,b\}$ and agent 2 receives the bundle $\{c,d\}$, or the allocation where agent 1 receives the bundle $\{b,d\}$ and agent 2 receives the bundle $\{a,c\}$, each with equal probability.}

The result of this process is ex-ante EF by the properties of the eating procedure \cite{BM01}. 
Additionally, \citet{A20} showed that the outcome is ex-post EF1.
However, neither the algorithm of~\citet{A20} (described above) nor the algorithm of~\citet{FSV20} provides any EFX guarantees, as they always return {\em balanced} allocations, where {all agents receive the same number of items.}
In general, there may be no balanced allocations that satisfy $\beta$-EFX for any $\beta > 0$.
This is demonstrated in our running example, where in any balanced allocation, the agent who gets the apple also gets another item, which inevitably violates the $\beta$-EFX condition {for sufficiently small $\varepsilon$}. 

\section{Omitted Propositions and Examples}

\subsection{The Analysis of Our Main Algorithm is Tight}

In this section, we show that the analysis of our main algorithm is tight. We also show that the simple modification of selecting the unenvied agent uniformly at random does not give any better guarantees. We use the following example to demonstrate these facts.

\begin{example}\label{ex:tight_analysis}
    Let $n, k, \varepsilon > 0$.
    Consider an instance with $n$ agents and $n+k$ items, {with the following additive valuations:}
    \begin{center}
        \begin{tabular}{||c | c c c c c | c c c ||} 
         \hline
         & $1$ & $2$ & $\ldots$ & $n-1$ & $n$ & $n+1$ & $\ldots$ & $n+k$ \\ [0.5ex] 
         \hline\hline
         $v_1$ & $1+n\varepsilon$ & $1+(n-1)\varepsilon$ & $\ldots$ & $1+2\varepsilon$ & $1+\varepsilon$ & $1$ & $\ldots$ & $1$ \\ 
         $v_2$ & $1+n\varepsilon$ & $1+(n-1)\varepsilon$ & $\ldots$ & $1+2\varepsilon$ & $1+\varepsilon$ & $1$ & $\ldots$ & $1$ \\ 
         $\vdots$ & $\vdots$ & $\vdots$ & $\ddots$ & $\vdots$ & $\vdots$ & $\vdots$ & $\ddots$ & $\vdots$ \\
         $v_{n-1}$ & $1+n\varepsilon$ & $1+(n-1)\varepsilon$ & $\ldots$ & $1+2\varepsilon$ & $1+\varepsilon$ & $1$ & $\ldots$ & $1$ \\ 
         \hline
         $v_{n}$ & $1$ & $1$ & $\ldots$ & $1$ & $1+\varepsilon$ & $1$ & $\ldots$ & $1$ \\
         \hline
        \end{tabular}
    \end{center}
\end{example}

During the eating procedure, each of the agents $1, \ldots, n-1$ consumes a $\frac{1}{n-1}$ fraction of each of the items $1, \ldots, n-1$, and agent $n$ consumes the entire item $n$. Therefore, after the first phase of the algorithm, there are exactly two agents that no one envies: some agent $i \in \{1, \ldots, n-1\}$ (who gets item $n-1$) and agent $n$ (who gets item $n$). 

Consider first the case with $k = 1$ to see that no deterministic tie-breaking rule for selecting the unenvied agent in line~\ref{line:unenvied} gives $\alpha$-EF for any $\alpha > 1/2$. Indeed, no matter which agent gets item $n+1$ during the second phase of the algorithm, any other agent will envy her with a factor arbitrarily close to $2$ for sufficiently small $\varepsilon$. 

We also show that selecting the unenvied agent uniformly at random from the set of unenvied agents does not give $\alpha$-EF for any $\alpha > 1/2$. Here, let us consider $n > k > 1$. Observe that agent $n$ is assigned one of the items $n+1, \ldots, n+k$ with probability $1-{1}/{2^k}$. Indeed, for every outcome of the first phase of the algorithm, the only outcome of the second phase in which $n$ does not get any of the additional items is where for each $1 \leq w \leq k$, the agent who got item $n-w$ in the first phase of the algorithm, gets item $n+w$ in the second phase, which happens with probability ${1}/{2^k}$. Therefore, agent $1$'s expected value for agent $n$'s bundle is $(1+\varepsilon) + (1-{1}/{2^k})$ while her expected value for her own bundle is $(1 + \frac{n+2}{2}\varepsilon) + \frac{k}{n}$. The claim follows by taking $k = \log n$ and $n \to \infty$.

\subsection{Ex-ante Guarantees of Random Serial Dictatorship}

In this section we establish an upper bound of $\frac{1}{\sqrt{2}}$ on the ex-ante EF guarantee of {random serial dictatorship}. This bound improves upon the previously known bound of $\frac{4}{5}$, from \citet{BM01}.

\begin{proposition}
\label{prop:upper-bound-rrr}
Random serial dictatorship is not ex-ante $\alpha$-EF for any $\alpha>1/\sqrt{2}$.
\end{proposition}

\begin{proof}
    Let $\varepsilon > 0$, and consider the following instance with $n$ additive agents and $n$ items, where $k$ is a parameter (to be determined later) and $p = k/n$. 
    \begin{center}
        \begin{tabular}{||c | c c | c c c | c c c ||} 
         \hline
         & $a$ & $b$ & $c_1$ & $\dots$ & $c_{k}$ & $d_1$ & \dots & $d_{n-k-2}$ \\ [0.5ex] 
         \hline\hline
         $v_1$ & $1+\varepsilon$ & $1$ & $0$ & $\dots$ & $0$ & $0$ & $\dots$ & $0$ \\ 
         \hline
         $v_2$ & $1$ & $1+\varepsilon$ & $0$ & $\dots$ & $0$ & $0$ & $\dots$ & $0$ \\
         \hline
         $v_3$ & $1$ & $0$ & $1+\varepsilon$ & $\dots$ & $1+\varepsilon$ & $\varepsilon$ & $\ldots$ & $\varepsilon$ \\ 
         $\vdots$ & $\vdots$ & $\vdots$ & $\vdots$ & $\ddots$ & $\vdots$ & $\vdots$ & $\ddots$ & $\vdots$ \\
         $v_n$ & $1$ & $0$ & $1+\varepsilon$ & $\dots$ & $1+\varepsilon$ & $\varepsilon$ & $\ldots$ & $\varepsilon$ \\ 
         \hline
        \end{tabular}
    \end{center}
Consider serial dictatorship with a uniformly random order. 
Observe that agents $3, \dots, n$ are identical, and let us denote by $t$ the $(k+1)^{\text{st}}$ of them to arrive. 
We proceed by a case analysis, based on the relative arrival order of agents $1$, $2$ and $t$.
For each one of these events, we write its probability, and the items chosen by agent 1 and by agent 2.
\begin{center}
\begin{tabular}{| c | r l | c | c |}
 \hline
Relative ordering & \multicolumn{2}{|c|}{Probability} & Item chosen by $1$ & Item chosen by $2$\\
 \hline\hline
 &&&&\\[-.7em]
$1,2,t$ &$\frac{1}{2}\frac{k+1}{n-1}\frac{k+2}{n}$&$ \approx \frac{p^2}{2}$ & $a$ & $b$ \\[.5em]
$2,1,t$ & $\frac{1}{2}\frac{k+1}{n-1}\frac{k+2}{n}$&$ \approx \frac{p^2}{2}$ & $a$ &$b$\\[.5em]
$1,t,2$ & $\frac{k+1}{n-1}\frac{n-k-2}{n}$&$ \approx p(1-p)$ &$a$ & $b$\\[.5em]
$2,t,1$ & $\frac{k+1}{n-1}\frac{n-k-2}{n} $&$\approx p(1-p)$& some $d_i$ & $b$\\[.5em]
$t,1,2$ & $\frac{1}{2}\frac{n-k-2}{n-1}\frac{n-k-1}{n} $&$\approx \frac{(1-p)^2}{2}$ & $b$ & some $d_i$\\[.5em]
$t,2,1$ & $\frac{1}{2}\frac{n-k-2}{n-1}\frac{n-k-1}{n}$&$ \approx \frac{(1-p)^2}{2}$ & some $d_i$ & $b$\\[.5em]
\hline
\end{tabular}
\end{center}
We observe that agent 2 is more likely to receive 
one of the top items {(item $a$ or $b$)}
than agent $1$, which creates ex-ante envy from agent $1$ to agent $2$. More precisely
\begin{align*}
\mathbb E[v_1(X_1)] &= (1+\varepsilon)\cdot\left(\frac{k+1}{n-1}\right) + 1\cdot \left(\frac{1}{2}\cdot\frac{n-k-2}{n-1}\cdot\frac{n-k-1}{n} \right) \approx \frac{1+p^2}{2}+\varepsilon p\\
\mathbb E[v_1(X_2)] &= 1\cdot \left(1-\frac{1}{2}\frac{n-k-2}{n-1}\frac{n-k-1}{n}\right) \approx 1-\frac{(1-p)^2}{2}.
\end{align*}
The ex-ante envy guarantee is upper-bounded by the ratio $\mathbb E[v_1(X_1)]/\mathbb E[v_1(X_2)]$. 
We take $n$ arbitrarily large, and $\varepsilon$ arbitrarily small, to simplify the expression, which gives
$$
\frac{\mathbb E[v_1(X_1)]}{\mathbb E[v_1(X_2)]} \approx \frac{1+p^2}{2-(1-p)^2}
$$
The above ratio reaches a minimum value of $1/\sqrt{2}$ when $p = \sqrt{2}-1$. More precisely, we choose $k=\lfloor n\cdot(\sqrt{2}-1)\rfloor$, which shows that random serial dictatorship is not ex-ante $\alpha$-EF for any $\alpha>1/\sqrt{2}$.
\end{proof}

\section[Omitted Proofs from Section 3.4]{Omitted Proofs from \Cref{sec:proof}}\label{sec:proofs}
\lemtech*
\begin{proof}
Properties (i) and (ii) follow directly from the definition. For property (iii), observe that for any random variables $x$ and $y$ that are positive on disjoint events, we have
$$
\mathbb{P}[x+y\geq t] = \mathbb{P}[x+y\geq t \text{ and } x>0] + \mathbb{P}[x+y\geq t \text{ and } y>0] = \mathbb{P}[x \geq t] + \mathbb{P}[y \geq t]
$$
for all $t > 0$, which concludes the proof.
\end{proof}

\lemextendability*
\begin{proof}
{Assume that all three properties hold. For all $t > \delta$, we have
$$
\mathbb{P}[z \geq t] = \mathbb{P}[z \geq t \text{ and } S\text{ holds}] + \mathbb{P}[z \geq t \text{ and } S\text{ does not hold}] = 0
$$
and thus $2\cdot \mathbb{P}[x\geq t] \geq \mathbb{P}[y\geq t]+\mathbb{P}[z\geq t]$. For all $0 < t \leq \delta$, we have
$$
2\cdot \mathbb{P}[x\geq t] = 2\cdot \mathbb{P}[S \text{ holds}] \geq \mathbb{P}[y > 0] + \mathbb{P}[z > 0] \geq \mathbb{P}[y\geq t] + \mathbb{P}[z\geq t],
$$
which concludes the proof.}
\end{proof}

\end{document}

%% file: figures/tradeoffs.tex
\hfill
\begin{minipage}{.47\textwidth}
\hspace{2.4cm} Two agents\\
\centering
\begin{tikzpicture}[scale=0.8]
\begin{axis}[
  width=\textwidth, height=\textwidth,
  clip = false,
  xtick={0,0.666,1},
  ytick={0,0.5,0.618,1},
  xticklabels = {$0$,$2/3$,$1$},
  yticklabels = {$0$,$1/2$,$\varphi-1$,$1$},
  xlabel={$\alpha$-EF},
  ylabel={$\beta$-EFX},
  xlabel style={at={(axis description cs:0.25,0)}},
  ylabel style={at={(axis description cs:0,0.25)}},
  xmin=0,  xmax=1, ymin=0,  ymax=1]
\addplot [draw=none, name path=f,domain=0.667:1] {(sqrt(4*x*x+1)-2*x+1)/(2*x)};
\path[name path=axis] (axis cs:0,1) -- (axis cs:1,1);
\tikzfillbetween[of=f and axis,on layer=main] {pattern=north east lines};
\draw (0.666,0.5) -- (0.666,1);
\draw[dashed] (0.666, 0) -- (0.666, 0.5);
\draw (0.666,0.5) -- (1,0.5);
\draw[dashed] (0,0.5) -- (0.666,0.5);
\fill[pattern=dots] (0,0) rectangle (0.666,1);
\fill[pattern=dots] (0.666,0) rectangle (1,0.5);

\node[fill=white] at (0.86,0.93)
{\footnotesize Prop.~\ref{prop:ub-subadditive}};
\node[fill=white] at (0.83,0.25)
{\footnotesize Prop.~\ref{prop:subadditive-2-agents-efx}};
\node[fill=white] at (0.33,0.75)
{\footnotesize Prop.~\ref{prop:subadditive-2-agents-ef}};
\node at (0.80,0.60)
{?};
\end{axis}
\end{tikzpicture}
\end{minipage}
\hfill
\begin{minipage}{.47\textwidth}
\hspace{2.4cm} $n$ agents\\
\centering
\begin{tikzpicture}[scale=0.8]
\begin{axis}[
  width=\textwidth, height=\textwidth,
  clip = false,
  xtick={0,0.5,0.666,1},
  ytick={0,0.5,0.618,1},
  xlabel={$\alpha$-EF},
  ylabel={$\beta$-EFX},
  xticklabels = {$0$,$1/2$,$2/3$, $1$},
  yticklabels = {$0$,$1/2$,$\varphi-1$,$1$},
  xlabel style={at={(axis description cs:0.25,0)}},
  ylabel style={at={(axis description cs:0,0.25)}},
  xmin=0,  xmax=1, ymin=0,  ymax=1]
\addplot [draw=none, name path=f,domain=0.667:1] {(sqrt(4*x*x+1)-2*x+1)/(2*x)};
\path[name path=axis] (axis cs:0,1) -- (axis cs:1,1);
\tikzfillbetween[of=f and axis,on layer=main] {pattern=north east lines};
\draw (0.5,0.5) -- (0.5,0);
\draw (0.5,0.5) -- (0,0.5);
\fill[pattern=dots] (0,0) rectangle (0.5,0.5);

\node at (0.66,0.66) {?};
\node[fill=white] at (0.25,0.25) {\footnotesize Thm.~\ref{thm:main_theorem}};
\node[fill=white] at (0.86,0.93)
{\footnotesize Prop.~\ref{prop:ub-subadditive}};
\end{axis}
\end{tikzpicture}
\end{minipage}
\hfill

%% file: figures/cycles.tex
\begin{minipage}{.3\textwidth}
\centering
\begin{tikzpicture}[thick,scale=3,>=stealth]
\node[draw,circle] (1) at (0,1) {1};
\node[draw,circle] (2) at (1,1) {2};
\node[draw,circle] (3) at (1,0) {3};
\node[draw,circle] (4) at (0,0) {4};
\draw[->] (1) to[bend left=20] (2);
\draw[->] (2) to[bend left=20] (3);
\draw[->] (3) to[bend left=20] (4);
\draw[->] (4) to[bend left=20] (1);
\draw[->] (4) to[bend left=20] (3);
\draw[->] (2) to[bend left=20] (1);
\draw[->] (1) to (3);
\end{tikzpicture}
\smallbreak
\smallbreak
\smallbreak
\smallbreak
(i) {Strongly connected graph}
\end{minipage}
\hfill
\begin{minipage}{.3\textwidth}
\centering
\begin{tikzpicture}[scale=3,thick]
\node[circle] (1) at (0,1) {$\frac{4}{13}$};
\node[circle] (2) at (1,1) {$\frac{3}{13}$};
\node[circle] (3) at (1,0) {$\frac{3}{13}$};
\node[circle] (4) at (0,0) {$\frac{3}{13}$};
\begin{scope}[-{Latex[width=8pt,length=4pt]}]
\draw[-{Latex[width=8pt,length=4pt]},double distance=4pt] (2) to[bend right=20] (1);
    \draw (2) to[bend right=20] (1);
\draw (3) to[bend right=20] (2);
\draw[double distance=4pt] (4) to[bend right=20] (3);
    \draw (4) to[bend right=20] (3);
\draw[double distance=2pt] (1) to[bend right=20] (4);
\draw (3) to[bend right=20] (4);
\draw[double distance=2pt] (1) to[bend right=20] (2);
\draw (3) to (1);
\end{scope}
\end{tikzpicture}
\smallbreak
\smallbreak
\smallbreak
(ii) Stationary distribution
\end{minipage}
\hfill
\begin{minipage}{.3\textwidth}
\centering
\begin{tikzpicture}[scale=3,thick]
\def\d{.05}
\def\triangleL{\scalebox{.6}{$\blacktriangleleft$}}
\def\triangleR{\scalebox{.6}{$\blacktriangleright$}}
\def\triangleU{\rotatebox{90}{\scalebox{.6}{$\blacktriangleleft$}}}
\def\triangleD{\rotatebox{-90}{\scalebox{.6}{$\blacktriangleleft$}}}
\def\triangleX{\rotatebox{135}{\scalebox{.6}{$\blacktriangleleft$}}}
\draw (\d,\d) to[bend right=20] node {\triangleL}
    (1-\d,\d) to[bend right=20] node {\triangleR} (\d,\d);
\draw (\d,1) to[bend right=20] node {\triangleL}
    (1-\d,1) to[bend right=20] node {\triangleR} (\d,1);
\draw (\d,1-\d) to[bend right=20] node {\triangleL}
    (1-\d,1-\d) to[bend right=20] node {\triangleR} (\d,1-\d);
\draw (-\d,-\d) to[bend right=20] node {\triangleL}
    (1+\d,-\d) to[bend right=20] node {\triangleU}
    (1+\d,1+\d) to[bend right=20] node {\triangleR}
    (-\d,1+\d) to[bend right=20] node {\triangleD} (-\d,-\d);
\draw (0,0) to[bend right=20] node {\triangleL}
    (1,0) to (1,2*\d) to node {\triangleX}
    (0,1-\d) to[bend right=20] node {\triangleD} (0,0);
\end{tikzpicture}
\smallbreak
(iii) Cycle distribution
\end{minipage}

%% file: algos/alg-fair_envy_cycles.tex
\begin{algorithm}[h]
\caption{Fair Envy Cycles.}\label{alg:fair_envy_cycles}
\begin{algorithmic}[1]
\vspace{0.1in}
\State\textbf{Input} A set $[m]$ of items, a set $[n]$ of agents, and a profile $(v_i)_{i\in[n]}$ of valuation functions.
\State\textbf{Output} A complete allocation $(X_1, \ldots, X_n)$ that is ex-ante $\frac{1}{2}$-EF, ex-post $\frac{1}{2}$-EFX and EF1.
\vspace{0.2cm}
\Statex -------------------------------------------------- \textbf{First Phase} --------------------------------------------------
\State $Z \gets \textsc{One-Step-PS}(m,n,(v_i)_{i\in[n]})$ \label{line:rr_1}
\State $(X^k,p^k)_{k\in[q]} \gets \textsc{Birkhoff-Rounding}(Z)$
\State $X \gets$ an integral allocation sampled from $(X^k,p^k)_{k\in[q]}$ \label{line:rr_2}
\vspace{0.2cm}
\Statex ------------------------------------------------ \textbf{Second Phase} ------------------------------------------------\While{there is an unallocated item}
\State $x \gets $ {an arbitrary unallocated item} \label{line:ec_1}
\If{there is an unenvied agent}
    \State $i \gets $ {an arbitrary} unenvied agent\label{line:unenvied}
    \State $X_i \gets X_i + x$\label{line:unenvied}
\Else
    \State $G \gets $ {an arbitrary} strongly connected component of 
    \State \hphantom{$G \gets $} the envy graph, {without any} incoming edges\label{line:cycles_1}
    \State $c \gets $ a cycle sampled from the distribution obtained
    \State \hphantom{$c \gets $} from Lemma~\ref{lem:cycle_distribution} applied to $G$
    \State reallocate the bundles along $c$\label{line:cycles_2}
\EndIf\label{line:ec_2}
\EndWhile
\Statex -----------------------------------------------------------------------------------------------------------------------
\State \textbf{return} $(X_1, \ldots, X_n)$
\end{algorithmic}
\end{algorithm}

%% file: figures/tree.tex
\begin{tikzpicture}[thick, scale=1.5]
    \fill[white!80!Blue] plot [smooth cycle,tension=1] coordinates {(-.5,0) (2.5,-2.5)  (2.5,2.5)};
    \fill[white!80!Blue] plot [smooth cycle,tension=1.5] coordinates {(3.5,2) (5,1.5)  (5,2.5)};
    \fill[white!80!Blue] plot [smooth cycle,tension=1] coordinates {(3.5,-1) (5.2,-.1)  (5.2,-1.9)};
    \begin{scope}[Blue]
    \node (PS) at (0.5,3) {Probabilistic Serial};
    \node (unenvied) at (3,3) {Item given to $j$};
    \node (cycle) at (4.4,.75) {Cycle elimination};
    \draw[->](cycle.-90) to[bend right=20] (4.2,-.2);
    \draw[->](unenvied.0) to[bend left=40] (4.5,2.6);
    \draw[->](PS.-90) to[bend right=20] (.7,2.2);
    \end{scope}
    {\footnotesize
    \node[draw,circle] (root) at (0,0) {root};
    \node[draw,circle] (r1) at (2,2) {~~~};
    \node[draw,circle] (r2) at (2,1) {~~~};
    \node[draw,circle] (r3) at (2,0) {~~~};
    \node[draw,circle] (r4) at (2,-1) {~~~};
    \node[draw,circle] (r5) at (2,-2) {~~~};
    \node[draw,circle] (v) at (4,2) {~~~};
    \node[draw,circle] (v1) at (5,2) {~~~};
    \node[draw,circle] (u) at (4,-1) {~~~};
    \node[draw,circle] (u1) at (5,-.4) {~~~};
    \node[draw,circle] (u2) at (5,-1) {~~~};
    \node[draw,circle] (u3) at (5,-1.6) {~~~};
    }
    \begin{scope}[densely dotted,rounded corners=.2cm]
    \draw (r1.90) -- (5.6,3) -- (6,3) -- (6,1) -- (5.6,1) -- (r1.-90);
    \draw (r4.90) -- (5.6,.4) -- (6,.4) -- (6,-2.4) -- (5.6,-2.4) -- (r4.-90);
    \end{scope}
    \draw (root) -- (r1);
    \draw (root) -- (r2);
    \draw (root) -- (r3);
    \draw (root) -- (r4);
    \draw (root) -- (r5);
    \draw (v) -- (v1);
    \draw[dashed] (u) -- (u1);
    \draw (u) -- (u2);
    \draw (u) -- (u3);
    \foreach \y in {-2.7,-2.5,...,2.9} {\draw (5.8,\y) circle (.05cm);}
    \foreach \y in {2.3,2.5} {\draw[dashed] (v1) -- (5.7,\y);}
    \foreach \y in {1.5,1.7,1.9,2.1} {\draw (v1) -- (5.7,\y);}
    \foreach \y in {-0.3,-0.1} {\draw[dashed] (u1) -- (5.7,\y);}
    \foreach \y in {-0.5} {\draw (u1) -- (5.7,\y);}
    \foreach \y in {-0.9,-0.7} {\draw[dashed] (u2) -- (5.7,\y);}
    \foreach \y in {-1.3,-1.1} {\draw (u2) -- (5.7,\y);}
    \foreach \y in {-1.5} {\draw[dashed] (u3) -- (5.7,\y);}
    \foreach \y in {-1.9,-1.7} {\draw (u3) -- (5.7,\y);}
    \node at (0,-3) {$\mathcal H_0$};
    \node at (2,-3) {$\mathcal H_1$};
    \node at (5.8,-3) {$\mathcal L$};
    \node[anchor=north west] at (6.2,2) {
    };
    \begin{scope}[xshift=4cm,yshift=2.1cm]
        \fill[white!80!Blue] plot [smooth cycle,tension=1] coordinates {(2.5,0) (4.2,-1)  (4.2,1)};
        \node[Blue] at (3.5,1.25) {Probabilistic Serial};
        \footnotesize
        \node[draw,circle] (rcopy) at (3,0) {$u$};
        \tiny
        \node[draw,circle] (r1copy) at (4,.8) {$u_1$};
        \node[draw,circle] (r2copy) at (4,.4) {$u_2$};
        \node[draw,circle] (r3copy) at (4,0) {$u_3$};
        \node[draw,circle,minimum width=.55cm] (r4copy) at (4,-.4) {~~~};
        \node at (4,-.32) {$\vdots$};
        \node[draw,circle] (r5copy) at (4,-.8) {$u_q$};
        \draw (rcopy) -- (r1copy);
        \draw (rcopy) -- (r2copy);
        \draw (rcopy) -- (r3copy);
        \draw (rcopy) -- (r4copy);
        \draw (rcopy) -- (r5copy);
        \foreach \n/\y in {r1copy/0.9,r2copy/0.5,r2copy/0.3,r3copy/0.1,r3copy/-.1,r3copy/-.3,r4copy/-.5,r5copy/-.9}
        {
        \fill[ForestGreen!50!white] (5.6,\y-.1) rectangle (6,\y+.1);
        \draw[] (\n) -- (5.7,\y);
        }
        \foreach \n/\y in {r1copy/1.1,r2copy/0.7,r4copy/-.7,r5copy/-1.1}
        {
        \fill[Goldenrod](5.6,\y-.1) rectangle (6,\y+.1);
        \draw[dashed] (\n) -- (5.7,\y);
        }
        \foreach \y in {-1.1,-.9,...,1.2} {\draw (5.8,\y) circle (.05cm);}
    \end{scope}
    \begin{scope}[xshift=4cm,yshift=-2cm]
        \node[Blue] at (3.5,2.75) {Item given to $j$};
        \footnotesize
        \fill[white!80!Blue] plot [smooth cycle,tension=1.5] coordinates {(2.5,2) (4,1.5)  (4,2.5)};
        \fill[Goldenrod] (5.6,2.2) rectangle (6,2.6);
        \fill[ForestGreen!50!white] (5.6,1.4) rectangle (6,2.2);
        \foreach \y in {1.5,1.7,...,2.3,2.5} {\draw (5.8,\y) circle (.05cm);}
        \node[draw,circle] (vcopy) at (3,2) {$u$};
        \node[draw,circle] (v1copy) at (4,2) {$u_1$};
        \draw (vcopy) -- (v1copy);
        \foreach \y in {2.3,2.5} {\draw[dashed] (v1copy) -- (5.7,\y);}
        \foreach \y in {1.5,1.7,1.9,2.1} {\draw (v1copy) -- (5.7,\y);}
    \end{scope}
    \begin{scope}[xshift=4cm,yshift=-1cm]
        \node[Blue] at (3.5,0.15) {Cycle elimination};
        \footnotesize
        \fill[white!80!Blue] plot [smooth cycle,tension=1] coordinates {(2.5,-1) (4.2,-.1)  (4.2,-1.9)};
        \fill[BrickRed!50!white] (5.6,0) rectangle (6,-.4);
        \fill[ForestGreen!50!white] (5.6,-.4) rectangle (6,-.6);
        \fill[Goldenrod] (5.6,-.6) rectangle (6,-1);
        \fill[ForestGreen!50!white] (5.6,-1) rectangle (6,-1.4);
        \fill[Goldenrod] (5.6,-1.4) rectangle (6,-1.6);
        \fill[ForestGreen!50!white] (5.6,-1.6) rectangle (6,-2);
        \foreach \y in {-1.9,-1.7,...,-.1} {\draw (5.8,\y) circle (.05cm);}
        \node[draw,circle] (ucopy) at (3,-1) {$u$};
        \node[draw,circle] (u1copy) at (4,-.4) {$u_1$};
        \node[draw,circle,minimum width=.7cm] (u2copy) at (4,-1) {};
        \node at (4,-.94) {$\vdots$};
        \node[draw,circle] (u3copy) at (4,-1.6) {$u_k$};
        \draw[dashed] (ucopy) -- (u1copy);
        \draw (ucopy) -- (u2copy);
        \draw (ucopy) -- (u3copy);
        \foreach \y in {-0.3,-0.1} {\draw[dashed] (u1copy) -- (5.7,\y);}
        \foreach \y in {-0.5} {\draw (u1copy) -- (5.7,\y);}
        \foreach \y in {-0.9,-0.7} {\draw[dashed] (u2copy) -- (5.7,\y);}
        \foreach \y in {-1.3,-1.1} {\draw (u2copy) -- (5.7,\y);}
        \foreach \y in {-1.5} {\draw[dashed] (u3copy) -- (5.7,\y);}
        \foreach \y in {-1.9,-1.7} {\draw (u3copy) -- (5.7,\y);}
    \end{scope}
\end{tikzpicture}

%% file: algos/det_envy_cycles.tex
\begin{algorithm}
\caption{Deterministic Envy Cycles.}\label{alg:det_envy_cycles}
\begin{algorithmic}[1]
\vspace{0.1cm}
\State\textbf{Input} A set $[m]$ of items, a set $[n]$ of agents, and a profile $(v_i)_{i\in[n]}$ of valuation functions.
\State\textbf{Output} A complete allocation $(X_1, \ldots, X_n)$ that is $\frac{1}{2}$-EFX and EF1.
\vspace{0.2cm}
\Statex -------------------------------------------------- \textbf{First Phase} --------------------------------------------------
\vspace{0.1cm}
\State $(X_1, \ldots, X_n) \gets$ {an arbitrary} weakly separated (Def.~\ref{def:weak_separation}) allocation with $|X_i| = 1$ for all $i$
\vspace{0.2cm}
\Statex ------------------------------------------------ \textbf{Second Phase} ------------------------------------------------\While{there is an unallocated item}
\State $x \gets $ an arbitrary unallocated item
\If{there is an unenvied agent}
    \State $i \gets $ an arbitrary unenvied agent
    \State $X_i \gets X_i + x$
\Else
    \State $i_1, \ldots, i_k \gets $ {an arbitrary} directed cycle in the envy graph
    \State $(Y_1, \ldots, Y_n) \gets$ $Y_{i_w} = X_{i_{w+1}}$ for all $1 \leq w < k$ 
    \State \;\;\;\;\;\;\;\;\;\;\;\;\;\;\;\;\;\;\;\;\;\;\; $Y_{i_k} \, = X_1$
    \State \;\;\;\;\;\;\;\;\;\;\;\;\;\;\;\;\;\;\;\;\;\;\; $Y_{i} \;\,\, = X_i$ for $i \neq \{i_1, \ldots, i_k\}$
    \State $(X_1, \ldots, X_n) \gets (Y_1, \ldots, Y_n)$
\EndIf
\EndWhile
\Statex -----------------------------------------------------------------------------------------------------------------------
\State \textbf{return} $(X_1, \ldots, X_n)$
\end{algorithmic}
\end{algorithm}

%% file: figures/eating.tex
\begin{minipage}{.45\textwidth}
\centering
\begin{tikzpicture}[xscale=1.5, yscale=.5]
\node at (-.5,6.5) {\includegraphics[width=.8cm]{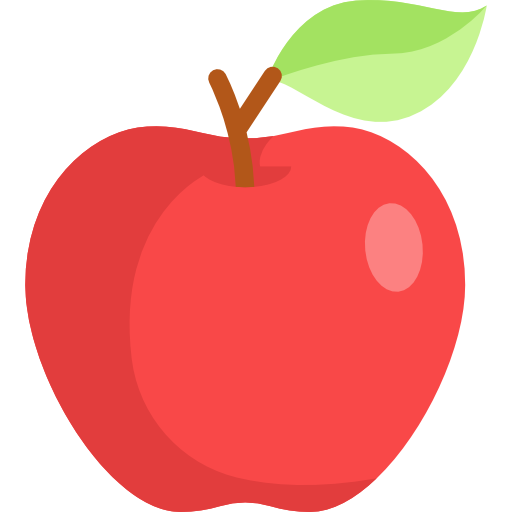}};
\node at (-.5,4.5) {\includegraphics[width=.8cm]{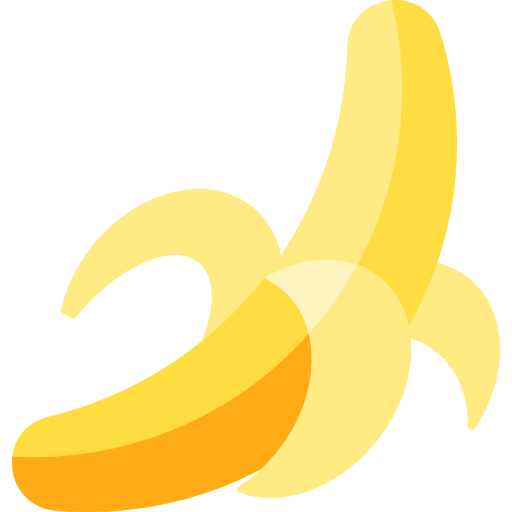}};
\node at (-.5,2.5) {\includegraphics[width=.8cm]{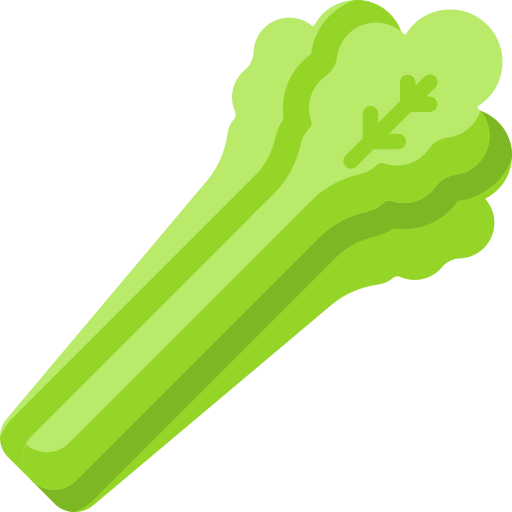}};
\node at (-.5,0.5) {\includegraphics[width=.8cm]{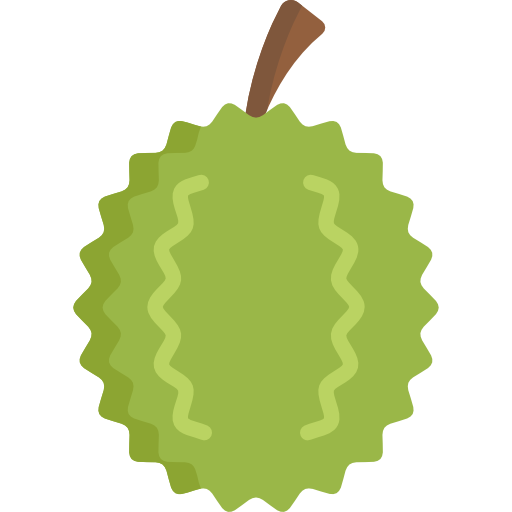}};
\filldraw[blue!15!white] (0,6) rectangle (1,7);
\filldraw[red!15!white] (0,2) rectangle (1,3);
\filldraw[blue!15!white] (0,4) rectangle (1,5);
\filldraw[red!15!white] (1,6) rectangle (2,7);
\draw (0,0) rectangle (2,1);
\draw (0,2) rectangle (2,3);
\draw (0,4) rectangle (2,5);
\draw (0,6) rectangle (2,7);
\draw (1,6) -- (1,7);
\draw (1,4) -- (1,5);
\draw (1,2) -- (1,3);
\node at (0.5,6.5) {$1_1$};
\node at (1.5,6.5) {$2_1$};
\node at (0.5,4.5) {$1_1$};
\node at (0.5,2.5) {$2_1$};
\end{tikzpicture}
\smallbreak
(a) After one unit of time
\end{minipage}
\hfill
\begin{minipage}{.45\textwidth}
\centering
\begin{tikzpicture}[xscale=1.5, yscale=.5]
\node at (-.5,6.5) {\includegraphics[width=.8cm]{figures/images/apple.png}};
\node at (-.5,4.5) {\includegraphics[width=.8cm]{figures/images/banana.png}};
\node at (-.5,2.5) {\includegraphics[width=.8cm]{figures/images/celery.png}};
\node at (-.5,0.5) {\includegraphics[width=.8cm]{figures/images/durian.png}};
\filldraw[blue!15!white] (0,6) rectangle (1,7);
\filldraw[red!15!white] (0,2) rectangle (1,3);
\filldraw[blue!15!white] (0,4) rectangle (1,5);
\filldraw[red!15!white] (1,6) rectangle (2,7);
\filldraw[blue!15!white] (0,0) rectangle (1,1);
\filldraw[red!15!white] (1,0) rectangle (2,1);
\filldraw[blue!15!white] (1,4) rectangle (2,5);
\filldraw[red!15!white] (1,2) rectangle (2,3);
\draw (0,0) rectangle (2,1);
\draw (0,2) rectangle (2,3);
\draw (0,4) rectangle (2,5);
\draw (0,6) rectangle (2,7);
\draw (1,6) -- (1,7);
\draw (1,4) -- (1,5);
\draw (1,2) -- (1,3);
\draw (1,0) -- (1,1);
\node at (0.5,6.5) {$1_1$};
\node at (1.5,6.5) {$2_1$};
\node at (0.5,4.5) {$1_1$};
\node at (1.5,4.5) {$1_2$};
\node at (0.5,2.5) {$2_1$};
\node at (1.5,2.5) {$2_2$};
\node at (0.5,0.5) {$1_2$};
\node at (1.5,0.5) {$2_2$};
\end{tikzpicture}
\smallbreak
(b) After two units of time
\end{minipage}